\documentclass[a4paper, 11pt]{article}
\usepackage{fullpage}
\usepackage{amsmath, amssymb,amsthm}
\usepackage{color,graphicx}
\usepackage[noend,ruled,vlined, linesnumbered]{algorithm2e}
\usepackage{tikz}
\usepackage{pgfplots}
\newcommand{\cF}{\ensuremath{\mathcal{F}}}
\newcommand{\cN}{\ensuremath{\mathcal{N}}}
\newcommand{\cS}{\ensuremath{\mathcal{S}}}
\newcommand{\cG}{\ensuremath{\mathcal{G}}}

\newcommand{\R}{\ensuremath{\mathbb{R}}}
\newcommand{\N}{\ensuremath{\mathbb{N}}}
\newcommand{\A}{\ensuremath{\mathbb{A}}}

\newcommand{\deconv}{\oslash}

\newcommand{\ba}{\boldsymbol{\alpha}}

\newcommand{\SD}{{\mathrm{SD}}}
\newcommand{\TD}{{\mathrm{TD}}}
\newcommand{\TDA}{{\mathrm{AG}}}
\newcommand{\TDC}{{\mathrm{2S}}}

\newcommand{\bb}{{\mathbf{b}}}

\newcommand{\sint}[1]{\mathbb{N}_{#1}}
\newcommand{\fl}{\mathrm{Fl}}
\newcommand{\path}{\rightsquigarrow}
\newcommand{\bu}{\bullet}
\DeclareMathOperator{\start}{start}

\newtheorem{theorem}{Theorem}
\newtheorem{lemma}{Lemma}
\newtheorem{definition}{Definition}
\newtheorem{example}{Example}

\newtheorem{corollary}{Corollary}

\begin{document}

\title{Stability and  performance guarantees in networks with cyclic dependencies}

\author{Anne Bouillard\\
Nokia Bell Labs France\\
Email: Anne.Bouillard@nokia-bell-labs.com
}

\maketitle

\begin{abstract}

  With the development of real-time networks such as reactive embedded
  systems, there is a need to compute deterministic performance
  bounds. This paper focuses on the performance guarantees and
  stability conditions in networks with cyclic dependencies in the
  network calculus framework. We first propose an algorithm that
  computes tight backlog bounds in tree networks for any set of flows
  crossing a server. Then, we show how this algorithm can be applied to improve bounds from the literature fir any topology, including cyclic networks. In particular, we show that the ring is stable in the network calculus framework. 
\end{abstract}

\section{Introduction} 
With the development of critical embedded systems, it becomes a
necessity to compute worst-case performance guarantees. Network
calculus is a (min,plus)-based theory that computes global performance bounds from a local
description of the network. These performances are the maximum
backlog at a server of end-to-end delay of a flow. Examples of
applications are switched network~\cite{Cruz1995},
Video-on-Demand~\cite{MR2006}... More recently, it has been very
useful for analysis large embedded networks such as AFDX (Avionics
Full Duplex)~\cite{BNOT2010}.

In most applications, such as AFDX, only feed-forward topologies are
used. One reason is the difficulty of deriving good deterministic
performance bounds in networks with cyclic dependencies. However,
allowing cycles in networks would result in a better bandwidth usage and
more flexible communications \cite{AMFLRU16}. As a consequence, there is
a strong need to design efficient methods for computing precise deterministic bounds. 

Recent works (\cite{BJT2010a,BN15,BS2016}) have focused on
computing tight performance bounds in feed-forward networks, but the
stability of a network is still an open problem in network calculus.

The most classical method for computing performance guarantees in
cyclic networks is to use the {\em fix-point} or {\em stopped-time}
method. It has first been presented in~\cite{Cruz1991b}. 
A sufficient condition for stability is obtained as the existence of a fix point in an equation derived from the network description. 

The theoretical aspects of the stability in deterministic queueing
networks have also been studied in the slightly different model named
{\em adversarial queueing network} (see \cite{BKRSW2001} for a precise
presentation). An injection rate per server (the rate at which data
crossing this server is sent into the network) is given instead of one
arrival rate per flow. Stability is stated in function of the topology as a minor-exclusion conditions in~\cite{ABS2004}, in function of a service
policy in~\cite{AAFLLK2001} or of the injection rate in~\cite{LPR2004}.

Fewer works concern the stability in Network calculus. The most
classical result is the stability of the ring, which is proved
in~\cite{TG1996} for work-conserving links and generalized
in~\cite{LT2001}. Instability results are provided
in~\cite{And2007, And2000}. In~\cite{And2007}, the authors even show
that the FIFO policy can be unstable at arbitrary small utilization
rates.  Some works, such as ~\cite{RL2008}, have focused on finding sufficient condition for
the stability in FIFO networks. 
 
Another direction of research has consisted in breaking the cyclic
dependencies in order to ensure the stability. 
Removing arcs could disconnect the network, but forbidding some paths
of length 2, as in the {\em turn-prohibition} method,
\cite{SKZ2002,PSKL2004}, can ensure both stability and connectivity.

\paragraph{Contributions}
In this paper, we study the problem of stability in networks with
cyclic dependencies in the network calculus formalism, using recent
the developments in~\cite{BN15} for tandem networks.  
Our main contributions are the following:

\begin{itemize}
\item we generalize the recent algorithm of~\cite{BN15}, that computes exact worst-case delays in a tandem network. We adapt it to compute the worst-case backlog of a server for any subset of flows crossing that server in tree networks. As a matter of fact, the algorithm in~\cite{BN15} can be deduced from this new algorithm, while the reverse is not true. As a by-product, we improve the results of~\cite{BS15} about sink-tree networks;
\item this new algorithm is used to compute new stability sufficient conditions
  in networks as mathematical programs. In particular, we demonstrate the stability of the ring in the network calculus framework. A weaker result had already been proved in ~\cite{LT2001} and~\cite{TG1996}, but our weaker assumption close a long-standing conjecture. 
\end{itemize}

  \bigskip

  The rest of the paper is organized as follows: in
  Section~\ref{sec:framework}, we recall the network calculus
  basics. Then in Section~\ref{sec:stopped}, we present generic
  mathematical programming methods to compute sufficient condition for
  the stability of networks.  Next in Section~\ref{sec:exact}, we give
  our algorithm that computes exact worst-case backlog in tree
  networks. We finally combine these results to compute new stability
  sufficient conditions in Section~\ref{sec:stability}, and we compare
  them through numerical experiments in Section~\ref{sec:numerical}.

\section{Network calculus framework and model}
\label{sec:framework}
We denote by $\N$ the set of non-negative integers $\{0,1,\ldots\}$
and for all $n\in\N$, we set $\sint{n}=\{1,\ldots,n\}$. For $x\in\R$, we
set $(x)_+ = \max(x,0)$.
We write~$\R_+$ for the set of non-negative reals.

While our model is in line with the standard definitions of networks calculus,
we present only the material that is needed in this paper. 
A more complete presentation of the
network calculus framework can be found in the reference
books~\cite{LT2001,Chang2000}. 

\subsection{Flow and server model}

\paragraph{Data flows}
Flows of data are represented by non-decreasing
and left-continuous functions that model the cumulative
processes. More precisely, if $A$ represents a flow at a certain point
in the network, $A(t)$ is the amount of data of this flow that crossed
this point in the time interval $[0,t)$, with the convention $A(0)=0$.  More
formally, let \begin{multline*} \cF =
\{f:\R_+ \rightarrow \R~|~f(0)=0, f\text{~non-decreasing }
\text{and left-continuous} \}.
\end{multline*}

A system $\cS$ is a non-deterministic relation between input
and output flows, where the number of inputs is the same as the
number of outputs: $\cS \subseteq \cF^m\times \cF^m$ and there is a
one-to-one relation between the inputs and the outputs of the system,
such that to each input flow corresponds one and only one output flow
that is causal -- no data is created or lost inside the system -- meaning
that for $((A_i)_{i=1}^m,(B_i)_{i=1}^m) \in \cS$, $\forall i
\in \sint{m}$, $A_i\geq B_i$. The vector
$((A_i)_{i=1}^m,(B_i)_{i=1}^m)$ is an {\em (admissible) trajectory} of
$\cS$ if $((A_i)_{i=1}^m,(B_i)_{i=1}^m) \in \cS$.
If $m=1$, i.e., the system has exactly one incoming and one outgoing flow, then we will also refer to it 
as a {\em server}.

\paragraph{Arrival curves} The notion of arrival curve is quite
simple: The amount of data that arrived during an interval of time is
a function of the length of this interval. More formally, let
$\alpha\in\cF$. A flow is constrained by the arrival curve $\alpha$,
or is $\alpha$-constrained if $$\forall s,t\in\R_+ \text{ with }s\leq t\colon
\quad A(t) - A(s) \leq \alpha(t-s).$$ 

A typical example of such arrival curve are the pseudo-linear 
{\em token-bucket}
functions: $\alpha_{b,r}:0\mapsto 0;~t\mapsto b + r
t$, if $t>0$. 
The burstiness parameter~$b$ can be interpreted as the maximal amount
of data that can arrive simultaneously and the rate~$r$ as a maximal
long-term arrival rate.
\paragraph{Service curves}
The role of a service curve is to constrain the relation between the
input of a server and its output.  Let $A$ be an cumulative arrival
process to a server and $B$ be its cumulative departure process. Several
types of service curves have been defined in the literature, and the
main types are the simple and strict service curves, which
we now define.

We say that $\beta$ is a {\em simple service curve} for a server $\cS$ if
$$\forall (A,B)\in\cS\colon \quad A\geq B\geq A* \beta,$$ with the convolution
$A * \beta(t) =\inf_{0\leq s\leq t}A(s) + \beta(t-s)$. 

An interval $I$ is a {\em backlogged period} for $(A,B)\in\cF\times\cF$ if
$\forall u\in I$, $A(u) > B(u)$. The start of the
backlogged period of an instant $t\in\R_+$ is $\start(t) = \sup\{u\leq t \mid A(u) =
B(u)\}$. As both $A$ and $B$ are left-continuous, we have $A(\start(t))
=B(\start(t))$, and $(\start(t),t]$ is always a backlogged period.

We say that~$\beta$ is a {\em strict service curve} for server~$\cS$ if
\begin{equation}
\label{eq:strict}
\begin{split}
& \forall (A,B) \in \cS\colon\ A\geq B \quad \text{ and } \\ & \forall \mathrm{~backlogged~periods~} (s,t]\colon\ 
B(t)-B(s) \geq \beta(t-s).
\end{split}
\end{equation}
We define $\cS(\beta)$ as the set of functions $(A,B)$ satisfying Equation~\eqref{eq:strict}.
The name strict service curves implies a difference to simple service
curves.  Works on the comparisons between the different types of
service curves can be found in~\cite{LT2001, BJT2010b}. In this
article, we will mainly deal with strict service curves, but will make use of
the convolution $A*\beta$ used to define simple service curves. 

A typical example of a service curve are  the {\em rate-latency}
functions: $\beta_{R,T}:t\mapsto R\cdot(t-T)_+$, where $T$ is the
latency until the server has to become active and $R$ is its minimal
service rate after this latency. 

Note that a server $\cS$ may not be deterministic, as the function
$\beta$ only corresponds to a guarantee on the service offered.  Among
 this non-determinism, we will focus on two modes of operation:

 \begin{itemize}
 \item {\bf Exact service mode:} During a backlogged period $(s,t]$, the service is {\em exact}
  if for all $u\in(s,t]$ we have $B(u) = A(\start(t)) + \beta(u-\start(t))$.
  (In this case, $\start(u) = \start(t)$.)
 \item {\bf Infinite service mode:} During an
  interval of time $(s,t]$, the service is {\em infinite} if
  $\forall u\in(s,t]\colon B(u) = A(u)$, i.e.,  the server serves all data
  instantaneously.
 \end{itemize}

For a system $\cS\in\cF^m\times \cF^m$ with $m$ inputs and outputs, we
say that $\cS$ offers a strict service curve $\beta$ if the
aggregate system is, that is, if $$\forall \left((A_i)_{i\in\sint{m}},(B_i)_{i\in\sint{m}}\right)\in\cS,~
\left(\sum_{i\in\sint{m}} A_i,\sum_{i\in\sint{m}} B_i\right)\in\cS(\beta).$$

We assume no knowledge about the service policy in this system (except
that it is FIFO per flow). 

\subsection{Performance guarantees}

\paragraph{Backlog} In this article, we focus on the worst-case backlog of a flow or a set
of flow at a given server of a network. Let $(A,B) \in\cF^2$ be an
admissible trajectory of a server.  The backlog of the server at time
$t$ is $b(t) = A(t) - B(t)\}$. The worst-case backlog is then
$b_{\max}=\sup_{t\geq 0} b(t)$. Graphically this is the maximal
vertical distance between $A$ and $B$.

We denote  $\ell(t) = t-\start(t)$, the length of the backlogged period upt to $t$ and $\ell_{\max} = \max_t \ell(t)$, the maximum length of a backlogged period. 

We denote $b_{\max}(\alpha,\beta)$
(resp. $\ell_{\max}(\alpha,\beta)$) the maximum backlog (resp. the
maximum length of a backlogged period) that can be obtained for a flow
that is $\alpha$-constrained crossing $\cS(\beta)$. For example, we have
\begin{itemize}
	\item $b_{\max}(\gamma_{b,r} ,\beta_{R,T}) =  b+rT $ if $r\leq R$, $=+\infty$ otherwise;
	\item $\ell_{\max}(\gamma_{b,r} ,\beta_{R,T}) = \frac{b + RT}{R-r})$
	if $r< R$, $=+\infty$ otherwise.
\end{itemize}

For a system with $m$ inputs and outputs, it is also possible to
compute the maximal backlog for a set of flows crossing this server. If
$I\subset \N_m$, the backlog of flows in $I$ at time $t$ is $b_I(t) =
\sum_{i\in I}F^{(in)}_i(t) - \sum_{i\in I}F^{(out)}_i(t)$. If the
system offers a strict service curve $\beta$ and flow $i$ is
$\alpha_i$-constrained, then $$b_I(t) \leq b_{I,\max} = (\sum_{i\in I}\alpha_i)
\deconv (\beta - \sum_{j\notin I} \alpha_j)_+(0),$$ where $f\deconv g(t)
= \sup_{u\geq 0} f(t+u)-g(u)$ is the (min,plus)-deconvolution. 
In the case of leaky-bucket arrival curves and rate latency service curve, 
\begin{equation}
\label{eq:BI}b_{I,\max} = b_I + \frac{r_I}{R-r_{\overline{I}}}(b_{\overline{I}} + r_{\overline{I}} T) + r_{{I}} T,
\end{equation}
with $r_I = \sum_{i\in I}r_i$, and similarly for $r_{\overline{I}}$, $b_I$ and $b_{\overline{I}}$.

\paragraph{Stability} We will also be interested in the stability of a network. Let us first
define the stability for server:

\begin{definition}[Server stability]
	Consider a server offering a strict service curve $\beta$ and a flow
	crossing it, with arrival curve $\alpha$.
	\begin{itemize}
		\item This server is said {\em unstable} if its worst-case backlog
		is unbounded;
		\item This server is said {\em critical} if its worst-case backlog
		is bounded, but the lengths of its backlogged periods are not
		bounded;%
		\item This server is said {\em stable} if the length of its backlogged periods is  bounded.
	\end{itemize}
\end{definition}

If the service curve is rate-latency $\beta_{R,T}$ and the
arrival curve token-bucket $\gamma_{b,r}$, then a server is
unstable if $R<r$, critical is $R=r$ and stable if
$R>r$.

Note that this definition only involves $r$ and $R$. The stability is
insensitive to $b$ and $T$, that only influence the size of the
backlog and backlogged period. 

\subsection{Network model}
Consider a network composed of $n$ servers numbered from 1 to $n$ and
crossed by $m$ flows named $f_1,\ldots,f_m$, such that
\begin{itemize}
	\item each server $j$ guarantees a rate-latency strict service curve
	$\beta^{(j)} = \beta_{R_j,T_j}$;
	\item each flow crosses the network along a path $\pi= \langle
	\pi_i(1),\ldots,\pi_i(\ell_i)\rangle$, where $\ell_i\geq 1$ is the
	length of the path. Each flow is constrained by the arrival curve
	$\alpha_i = \gamma_{b_i,r_i}$.
\end{itemize}

For a server $j$, we define $\fl(j) = \{i~|~\exists \ell,~\pi(\ell) = j\}$ the set of flows crossing server $j$.

We denote by $\cN$ this network. Its induced graph is the directed graph whose vertices are the servers and the set of arcs is 
$$\A = \{(\pi_i(k),\pi_i(k+1))~|~i\in\sint{m}, k\in\sint{\ell_i-1}\}.$$

We assume, without loss of generality, as we will focus on the
performances in one server, that the network is connected has a unique
final strictly connected component. Moreover, we assume that for all
$j\in\sint{n-1}$, there exists $j'>j$ such that $(j,j')\in \A$ (up to
renumbering the servers, this is also without loss of generality).

\paragraph{Classes of networks:}
\begin{itemize}
	\item if $\A = \{(j,j+1)~|~j\in\sint{n-1}\}$, then the network is called an {\em tandem networks};
	\item if the output degree of each vertex except node $n$ is 1, then the network is called a {\em tree network};
	\item if the graph network has no cycle, the network is called {\em feed-forward};
	\item otherwise, it has {\em cyclic dependencies}.
\end{itemize}

\paragraph{Network stability}
\begin{definition}[Local stability]
	Consider a network $\cN$. It is said locally stable if all its servers are stable using the initial arrival curves:  $\forall j\in\sint{n}$, 
	$$ \ell_{\max}(\sum_{i\ni j} \alpha_i , \beta^{(j)})<\infty.$$
\end{definition}

\begin{definition}[Global stability]
	A network  is {\em globally stable} if for all its servers, the length of the maximal backlogged period is bounded.
\end{definition}

We call the {\em linear model} when arrival curves are leaky-bucket and the service curve rate-latency. 
 Our aim is to give sufficient properties for the
global stability for networks with cyclic dependencies (the underlying
graph has cycles). 

\begin{lemma}
	If a network is globally stable, then it is locally stable.
\end{lemma}

\begin{proof}
	We prove this by contra-position: suppose that the network is not
	locally stable. Then, there exists a server $j$ that is not stable
	considering the original arrival process. Consider the following
	trajectory: every server acts as an infinite server except server
	$j$. Then the arrival processes into server $j$ are exactly those
	that are injected into the network, and server $j$ is not stable:
	the length of its backlogged period cannot be bounded. But they
	cannot either in this behavior of the network. So the network is not
	globally stable.
\end{proof}

\renewcommand{\r}{\mathrm{r}}

\section{Sufficient conditions of network stability}
\label{sec:stopped}
When the network is feed-forward, it is possible to compute the
performance of the network by 
\begin{enumerate}
	\item applying Equation~\eqref{eq:BI} at every server
	and for $I=\{i\}$ for each flow $f_i$ crossing that server and
	\item propagating the constraints in the topological order of the servers.
\end{enumerate}

This is not possible when there are cyclic dependencies,
because of the inter-dependencies between the backlogs computed
when applying Equation~\eqref{eq:BI}. 

The fix-point method is a generic method to compute performance
guarantees in networks with cyclic dependencies. The main idea is to
compute, for each server and each flow crossing it, an output arrival
curve that depends on the input arrival curves at that server. A
system of equations is then obtained and the solution, if it exists,
gives output arrival curves for each flow after each
server it crosses. This approach has been described for leaky bucket
arrival curves and rate-latency service curves in ~\cite{LT2001}. The
method for proving this approach is the {\em stopped-time} method.

In this section, we generalize the stopped-time method described in~\cite{LT2001} in four directions:

\begin{itemize}
\item it can be applied to any arrival and service curve;
\item it can be applied to any feed-forward decomposition of the network
  rather than a for each server individually;
\item it can be applied to any group of flows rather than considering each flow individually;
\item it can handle different combinations of the two cases above.
\end{itemize}

Instead of computing a fix-point, we rather present our results as the solution of an optimization problem, that enables us to include the computation of the network performance (worst-case backlog for example) in problem. In the linear model, we obtain a linear program. 

With the three first points of generalization, it is also possible to solve a fix-point problem in order to obtain stability sufficient conditions. But the fourth point has no natural formalization into a fix-point problem.

We first describe the transformation of the network and flow grouping before and
 computing performance bounds as an optimization problem.

\subsection{Network transformation}
\label{ssec:fft}

\subsubsection{Feed-forward decomposition}
Let $\cN$ be a network and $G_{\cN}$ be its induced graph. This graph can be transformed into an acyclic graph by
removing a set of arc $\A^{\mathrm{r}}\subseteq \A$.
\begin{example}
\label{ex:toy}The toy network of Fig.~\ref{fig:toy} can be transformed into an acyclic network by removing arc
  $(4,2)$. In the next section, we will see that [erformance bounds can be efficiently computed in tree topologies. Such a decomposition can be obtained with $\A^\r = \{(4,2),(2,1)\}$. With $\A^\r=\A$, all arcs are removed, and we obtain a graph with isolated nodes only. 
 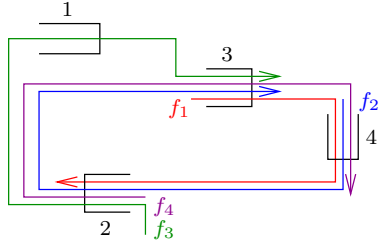
\begin{figure}[htbp]
    \centering
    \begin{picture}(0,0)%
\includegraphics{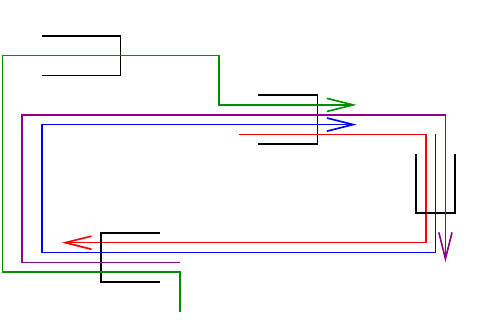}%
\end{picture}%
\setlength{\unitlength}{4144sp}%
\begingroup\makeatletter\ifx\SetFigFont\undefined%
\gdef\SetFigFont#1#2#3#4#5{%
  \reset@font\fontsize{#1}{#2pt}%
  \fontfamily{#3}\fontseries{#4}\fontshape{#5}%
  \selectfont}%
\fi\endgroup%
\begin{picture}(2218,1461)(4849,-4072)
\put(5176,-2716){\makebox(0,0)[lb]{\smash{{\SetFigFont{8}{9.6}{\rmdefault}{\mddefault}{\updefault}{\color[rgb]{0,0,0}1}%
}}}}
\put(6121,-2986){\makebox(0,0)[lb]{\smash{{\SetFigFont{8}{9.6}{\rmdefault}{\mddefault}{\updefault}{\color[rgb]{0,0,0}3}%
}}}}
\put(5401,-4021){\makebox(0,0)[lb]{\smash{{\SetFigFont{8}{9.6}{\rmdefault}{\mddefault}{\updefault}{\color[rgb]{0,0,0}2}%
}}}}
\put(6976,-3481){\makebox(0,0)[lb]{\smash{{\SetFigFont{8}{9.6}{\rmdefault}{\mddefault}{\updefault}{\color[rgb]{0,0,0}4}%
}}}}
\put(5716,-4021){\makebox(0,0)[lb]{\smash{{\SetFigFont{8}{9.6}{\rmdefault}{\mddefault}{\updefault}{\color[rgb]{0,.56,0}$f_3$}%
}}}}
\put(5716,-3886){\makebox(0,0)[lb]{\smash{{\SetFigFont{8}{9.6}{\rmdefault}{\mddefault}{\updefault}{\color[rgb]{.56,0,.56}$f_4$}%
}}}}
\put(6931,-3256){\makebox(0,0)[lb]{\smash{{\SetFigFont{8}{9.6}{\rmdefault}{\mddefault}{\updefault}{\color[rgb]{0,0,1}$f_2$}%
}}}}
\put(5806,-3301){\makebox(0,0)[lb]{\smash{{\SetFigFont{8}{9.6}{\rmdefault}{\mddefault}{\updefault}{\color[rgb]{1,0,0}$f_1$}%
}}}}
\end{picture}%
    \caption{Toy network of Example~\ref{ex:toy}.}
    \label{fig:toy}
  \end{figure}
\end{example}

Note that this transformation is not unique, and finding the minimum
set of edges to remove is a NP-complete problem (it is the {\em
  Minimum feed-back arc set} problem in~\cite{GJ1979}).  The most
classical solution in the literature is to remove every arc -- each
server is analyzed in isolation -- but it might be a better choice to
remove fewer arcs, and obtain a tree for example, we will see later that these topologies can be easily analyzed.

We now modify the flows in accordance to the arcs that have been
removed: each flow $f_i$ is split into several flows $f_{i,1},
f_{i,2},..., f_{i,m_i}$ of respective paths in
$(\sint{n},\A-\A^\r)$, $\pi_{i,1} = \langle
\pi_i(1),...,\pi_i(k_1)\rangle$, $\pi_{i,2}=\langle
\pi_i(k_1+1),...,\pi_i(k_2)\rangle$,..., $\pi_{i,m_i}=\langle
\pi_i(k_{m_i}+1),...,\pi_i(\ell_i)\rangle$, where
$(\pi_i(k_j+1),\pi_i(k_j))\in\A^r$.  We denote $\ell_{i,k}$ the
length of $\pi_{i,k}$, and call $\cN^{\mathrm{FF}}$ the feed-forward network that is obtained.

\begin{example}
  \label{ex:toy2}
  In the toy example of Fig.~\ref{fig:toy}, if $\A^\r = \{(4,2),(2,1)\}$, 
 flow $f_1$ is split into $f_{1,1}$ and $f_{1,2}$ with
  respective paths $\langle 3,4\rangle$ and $\langle 2\rangle$,
  flow $f_2$ is split into $f_{2,1}$ and $f_{2,2}$, with respective paths
  $\langle 4\rangle$ and $\langle 2,3\rangle$ and flow $f_3$ is split into $f_{3,1}$ and $f_{3,2}$, with respective paths
  $\langle 2\rangle$ and $\langle 1,3\rangle$. Flow $f_4$ remains unchanged ($f_{4,1} = f_4$). The result of this decomposition is depicted Fig.~\ref{fig:toy-unfold}.
\begin{figure}[htbp]
  \centering
  \begin{picture}(0,0)%
\includegraphics{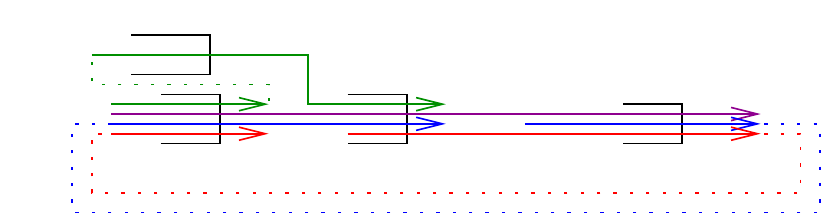}%
\end{picture}%
\setlength{\unitlength}{4144sp}%
\begingroup\makeatletter\ifx\SetFigFont\undefined%
\gdef\SetFigFont#1#2#3#4#5{%
  \reset@font\fontsize{#1}{#2pt}%
  \fontfamily{#3}\fontseries{#4}\fontshape{#5}%
  \selectfont}%
\fi\endgroup%
\begin{picture}(3762,972)(4441,-3583)
\put(5176,-2716){\makebox(0,0)[lb]{\smash{{\SetFigFont{8}{9.6}{\rmdefault}{\mddefault}{\updefault}{\color[rgb]{0,0,0}1}%
}}}}
\put(6121,-2986){\makebox(0,0)[lb]{\smash{{\SetFigFont{8}{9.6}{\rmdefault}{\mddefault}{\updefault}{\color[rgb]{0,0,0}3}%
}}}}
\put(7381,-3346){\makebox(0,0)[lb]{\smash{{\SetFigFont{8}{9.6}{\rmdefault}{\mddefault}{\updefault}{\color[rgb]{0,0,0}4}%
}}}}
\put(5266,-3391){\makebox(0,0)[lb]{\smash{{\SetFigFont{8}{9.6}{\rmdefault}{\mddefault}{\updefault}{\color[rgb]{0,0,0}2}%
}}}}
\put(4591,-3121){\makebox(0,0)[lb]{\smash{{\SetFigFont{8}{9.6}{\rmdefault}{\mddefault}{\updefault}{\color[rgb]{0,.56,0}$f_3$}%
}}}}
\put(6706,-3346){\makebox(0,0)[lb]{\smash{{\SetFigFont{8}{9.6}{\rmdefault}{\mddefault}{\updefault}{\color[rgb]{0,0,1}$f_2$}%
}}}}
\put(4591,-2896){\makebox(0,0)[lb]{\smash{{\SetFigFont{8}{9.6}{\rmdefault}{\mddefault}{\updefault}{\color[rgb]{0,.56,0}$f_{3,2}$}%
}}}}
\put(4456,-3436){\makebox(0,0)[lb]{\smash{{\SetFigFont{8}{9.6}{\rmdefault}{\mddefault}{\updefault}{\color[rgb]{1,0,0}$f_{1,2}$}%
}}}}
\put(4456,-3301){\makebox(0,0)[lb]{\smash{{\SetFigFont{8}{9.6}{\rmdefault}{\mddefault}{\updefault}{\color[rgb]{0,0,1}$f_{2,2}$}%
}}}}
\put(4456,-3166){\makebox(0,0)[lb]{\smash{{\SetFigFont{8}{9.6}{\rmdefault}{\mddefault}{\updefault}{\color[rgb]{.56,0,.56}$f_4$}%
}}}}
\put(5851,-3301){\makebox(0,0)[lb]{\smash{{\SetFigFont{8}{9.6}{\rmdefault}{\mddefault}{\updefault}{\color[rgb]{1,0,0}$f_1$}%
}}}}
\end{picture}%
    \caption{Feed-forward transformation of the toy network.}
    \label{fig:toy-unfold}
  \end{figure}
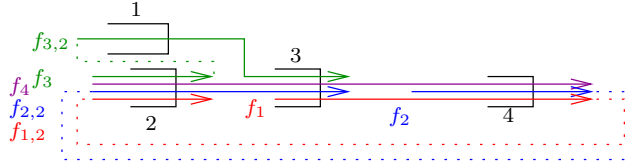
\end{example}

\subsubsection{Flow grouping}The second step is to group flows. Instead of dealing with flows
$f_{i,k}$ individually, we use a partition. Let
$S=\{(i,k)~|~i\in\sint{m}$,~$k\in\sint{m_i}\}$ be the set of flows
that have been created with the feed-forward transformation and
$S_1,\ldots,S_L \subseteq S$ with
$S_1\sqcup S_2 \sqcup \cdots \sqcup S_L = S$ (the symbol $\sqcup$ denotes the disjoint union).

Our goal is to compute $(\alpha^\ell)_{\ell\in\N_L}$ where
$\alpha^\ell$ is an arrival curves of the aggregation of the flows $(f_s)_{s\in S_\ell}$ in
$S_\ell$. We denote $f^{\ell}$ the aggregated flow.

This formulation might at first seem quite strange, but it takes several
interesting cases into account. 

\begin{example}
\label{ex:group}
  Consider again the toy example with the decomposition of
  Fig.~\ref{fig:toy-unfold}. There is a priori no use to group flows
  $f_{i,1}$, as their arrival curve is already known. For the other
  flows, there are two natural groupings. The first is to group
  individually: for all $\ell$, $S_\ell$ is a singleton. The second one is to group
  according to the removed arcs, and have the grouping
  $\{(1,2),(2,2)\}$ and $\{(3,2)\}$. Indeed, the arrival curve of
  $f_{1,2}+f_{2,2}$ can be less than the sum of the arrival curves of
  each individual flow, so hopefully, better performances can be
  computed with this decomposition.
\end{example}

\subsection{Stability and performances as an optimization problem}
If $\cN$ is stable, then there exists an arrival curve
$\alpha^\ell$ for the aggregated flow $f^\ell$. 

We make the following assumptions:

\begin{itemize}
\item[($\mathbf{A}_1$)] For all $\ell \in\N_L$, there exists a non-decreasing function $H_\ell:\cF^L \to \cF$ that computes an arrival curve in $\cN^{\mathrm{FF}}$ for the aggregation of flows $(f_{i,k-1})_{(i,k)\in S_\ell}$ at the end of their respective path $\pi_{i,k-1}$  in function of $\alpha^1,\ldots,\alpha^L$, the respective arrival curves of all aggregated flows $f^\ell$.
\item[($\mathbf{A}_2$)] There exists a non-decreasing function
  $G: \cF^L \to \R_+$ that computes a performance $P$ in $\cN$ as a  function of $(\alpha^\ell)_{\ell=1}^{L}$, arrival curves for the aggregate flows $f^\ell$.
\item[($\mathbf{A}_3$)] $(H_\ell)_{\ell=1}^{L}$ and $G$ implicitly depend on the
  arrival curves $(\alpha_i)_{i=1}^m$ and on the service curves
  $(\beta_j)_{j=1}^n$. These functions are assumed to be
  non-decreasing with $\alpha_i$, $i\in\N_m$ and
  non-increasing with $\beta_j$, $j\in\N_n$.
 \end{itemize}
Assumptions ($\mathbf{A}_1$) and ($\mathbf{A}_2$) are ensured when $\cN^{\mathrm{FF}}$ is feed-forward. It corresponds to choosing a method for computing performances in feed-forward networks. 
Assumption ($\mathbf{A}_3$) and the fact that $H_\ell$ and $G$ are non decreasing are made without loss of generality since with greater arrival
  curves and smaller service curves induce more admissible trajectories, hence larger worst-case performances.

Since the minimum of two arrival curves for a flow is also an arrival curve for that flow, when $\cN$ is stable, we can consider that $\alpha^\ell$ is the minimum arrival
curves for $f^{\ell}$. Then it holds $\forall \ell\in\N_L$, $\alpha^\ell \leq H_\ell(\ba)$, where $(\alpha^\ell)_{\ell=1}^{L} = \ba$. We can write this latter inequality as a vector expression:

\begin{equation}
\label{eq:fix-point}
\ba\leq H(\ba),
\end{equation}
where $H = (H_\ell)_{\ell=1}^L$.  The next theorem shows the reverse: if the solutions of Equation~\eqref{eq:fix-point} are bounded, then the system is stable. 

\begin{theorem}
  \label{th:stop}
  Set $\mathcal{C} = \{\ba~|~ \ba \leq H(\ba)\}$ be the set of
  solutions of Equation~\eqref{eq:fix-point}, and $\ba_0 =
  \sup\{\ba~|~\ba\in\mathcal{C}\}$.  If $\ba_0$ is finite, then $\cN$ is
  globally stable and for all $\ell\in \N_L$, $\alpha_0^{\ell}$ is an arrival
  curve for the aggregation of flows $f_i$ at the input of server $\pi_{i,k}(1)$ for $(i,k)\in S_\ell$.
\end{theorem}

The proof of this theorems follows exactly the same lines as the stopped-time method described in~\cite{LT2001} or~\cite{Chang2000}. 

\begin{proof}
   First, $\ba_0$ exists, as $\mathcal{C}$ contains a maximum element:
  if $\ba_1$ and $\ba_2$ are two elements in $\mathcal{C}$, then for all
  $\ell\in\N_L$, $\alpha_1^\ell \leq H_\ell(\ba_1) \leq H_\ell(\ba_1\lor \ba_2)$ and
  similarly, $\alpha_2^\ell \leq H_\ell(\ba_1\lor \ba_2)$, so
  $\alpha_1^\ell \lor\alpha_2^\ell \leq H_\ell(\ba_1\lor \ba_2)$, and $\ba_1\lor \ba_2
  \in \mathcal{C}$.
  
  We use the stopped-time method. Consider that the arrivals to the
  network stop at time $\tau>0$: for each flow $f_i$, an arrival curve
  is then $\alpha_i^\tau: t\mapsto \alpha_i(t\land \tau)$. The total
  amount of data for each flow $f_i$ is also bounded by
  $\alpha_i(\tau)$, so the network is globally stable.

Let $\ba^\tau = (\alpha^{\tau,\ell})_{\ell \in \N_L}$ be the family of the {\em
    minimal arrival curve} of the aggregated flows $f_i$ at server $\pi_{i,k}(1)$ for $(i,k)\in S_\ell$ then
$$\ba^{\tau} \leq H^{\tau}(\ba^{\tau}) \leq
H(\ba^{\tau}),$$
where $H^{\tau}$ is obtained the same way as $H$, but
replacing $\alpha_i$ by $\alpha_i^{\tau}$. The first inequality comes from the stability of $\cN$ for the stopped process, and the second from Assumption~($\mathbf{A}_3$).

For all $\tau>0$, $\ba^\tau \in\mathcal{C}$, so $\ba_0\geq \ba^{\tau}$,
and $\ba_0$ is a family of arrival curves that is valid for all $\tau>0$. Then it is valid for the whole unstopped process, and the system is stable if $\ba_0$ is finite. 
\end{proof}

\subsubsection{One-stage optimization problem}

We are now ready to give a mathematical programming problem  that computes worst-case performances upper bounds for arbitrary networks: 

\begin{equation}\label{1stage}\begin{array}{|l|}
  \hline
  \text{Maximize }G(\ba) \text{ such that }\ba \leq H(\ba).\\
  \hline
\end{array}
\end{equation}

\begin{theorem}
Under Assumptions ($\mathbf{A}_1$), ($\mathbf{A}_2$) and ($\mathbf{A}_3$), the solution of the optimization problem of Equation~\eqref{1stage} is an upper bound of the performance $P$ of the network. 
\end{theorem}
\begin{proof}
As $H_\ell$ and $G$ are non-decreasing,  $\ba^0$ maximizes $G$ among all the elements such that $\ba \leq H(\ba)$.
\end{proof}

This formulation is in fact equivalent to the {\em fix-point method}
that can be found in the literature, and that can be deduced from
Theorem~\ref{th:stop}. Indeed, the proof of that theorem ensures the
existence of a greatest fix-point. As $\ba_0$ is that fix-point, then
the performances can be directly computed as $G(\ba_0)$. 

This result is often used when there exists a unique fix-point to the
equation $\ba = H(\ba)$, which is also the largest solution of $\ba\leq H(\ba)$. Our formulation enables to apply the fix-point method in cases the uniqueness is not ensured. 

\subsubsection{Two-stage optimization problem}

The formulation as an optimization problem can be generalized, by
making advantage of two decompositions. For example, one could have a
feed-forward transformation of the network so that $\cN^{\mathrm{FF}}$
is a tree. Following example~\ref{ex:group}, there are two natural
ways to group flows. First considering singletons only, which
defines functions $G$ and $H_1$; second grouping flows according to
the arc that has been removed, which define functions $H_2$.

So the following mathematical program is obtained:

\begin{equation}\label{2stage}
\begin{array}{|rl|}
  \hline
  \text{Maximize } G(\ba)& \\ \text{ such that } &\ba\leq \ba_1,~ \ba \leq \ba_2\\
&\alpha_1^s \leq \alpha_2^\ell \leq \sum_{u\in S_\ell} \alpha_1^u,~\forall s\in S^\ell\\
  &\ba_1 \leq H_1(\ba_1),~\ba_2 \leq H_2(\ba_2)\\
  \hline
\end{array}
\end{equation}
\begin{theorem}
Under Assumptions ($\mathbf{A}_1$), ($\mathbf{A}_2$) and ($\mathbf{A}_3$), the solution of the optimization problem of Equation~\eqref{2stage} is an upper bound of the performance $P$ of the network. 
\end{theorem}
\begin{proof}
  Let $\ba_1$ be the vector of the smallest arrival curves for the individual flows, and $\ba_2$ for the aggregated flows. If $s\in S_\ell$, then $\alpha_s \leq \alpha^{\ell}$, as flow $f_s$ is part of the aggregated flow $f^\ell$, and $\alpha^\ell$ is less than the sum of the arrival curves of all the aggregated flows. Hence,  $\alpha_1^s \leq \alpha_2^\ell \leq \sum_{u\in S_\ell} \alpha_1^u$ is satisfied.
\end{proof}

The problem of Eq~\eqref{2stage} has no natural equivalent as a fix-point equation, and can be directly generalized for more than two stage, and different decompositions.
We will see that this optimization problem is slightly improved in the linear model.

\section{Worst-case backlog in tree networks}
\label{sec:exact}
In this section, we focus on tree networks and give an algorithm to
compute exact worst-case backlog in the linear model.
The
algorithm is a generalization of the one given in~\cite{BN15} with the
following differences:
\begin{enumerate}
\item our algorithm computes a worst-case backlog at a server;
\item it can be applied to compute the worst-case backlog at a server for any
  set of flows crossing this server;
\item it is valid for any tree topology.
\end{enumerate}
The two algorithms and their proof are based on the same ideas, so we skip the detailed proof here. The complete proof is in Appendix~\ref{appA}.

Let us first give some additional notations used in the algorithm to
describe a tree network. First, its induced graph is a tree
directed to vertex $n$, whose output degree 0. Each other vertex $j$
has output degree 1. We denote by $j^\bu$ its successor and assume
that $j<j^{\bu}$ and set $n^{\bu} = n+1$ by convention. The set
of predecessor of a vertex is ${}^{\bu}\!j = \{k~|~k^{\bu} =
j\}$. There exists at most one path between two vertices $j$ and $k$,
denoted ${j\rightsquigarrow k}$.
Finally, if there exists a path from $j$ to $k$,
${}^{\bu \!j}\!k$ is the predecessor of $k$ of this path.

Suppose that we are interested in computing the worst-case backlog at
server $n$ for some flows crossing it. We denote by $I\subseteq
\sint{m}$ those flows of interest. 

\begin{itemize}
\item $r_j^k = \sum_{i\in\fl(j)\setminus I, \pi_i(\ell_i)=k}r_i $ is the arrival rate at
  server~$j$ for all flows ending at server~$k$ and crossing server~$j$ that are not of interest;
\item $r^*_j = \sum_{i\in I\cap \fl(j)}r_i $ is the arrival rate of the flows of interests that cross server $j$.
\end{itemize}

\begin{algorithm}[htbp]
  \Begin{
    {$\xi_n^n\gets r^*_n / R_n-r^n_n$\;
    $Q = \mathtt{queue}({}^\bu\!n)$\;
    \While{$Q \neq \emptyset$}
    {
      $j=Q[0]$\;
      $k\gets n$\;
      \While{{ $\xi^k_{j^\bu}\!>\!(r^*_j \!+ \!\sum_{\ell\in k^{\bu}\path n} \!\xi_{j^{\bu}}^\ell r_j^\ell)/ (R_j-\sum_{\ell\in j \path k}r_j^\ell)$}} 
      {
         $\xi^k_j\gets \xi^k_{j^\bu}$\;
         $k\gets {}^{\bu j}\!k$\;
       }
       \For{$\ell$ {\bf from} $j$ {\bf to} $k$}
       {
         {$\xi^\ell_{j} \gets (r^*_j + \sum_{\ell'\in k^{\bu} \path n} \xi_{j^\bu}^{\ell'} r_j^{\ell'})/  (R_j-\sum_{\ell\in j \path k}r_j^{\ell'})$}\;
       }
       $Q \gets \mathtt{enqueue}(\mathtt{dequeue}(Q,j), {}^\bu \!j)$\;
     }
     \lFor {$j$ {\bf from} $1$ {\bf to} $n$}
     {
       $\rho_j\gets r^*_j+\sum_{\ell \in j \path n} \xi_j^\ell r_j^\ell$
     }
     \For {$i$ {\bf from} $1$ {\bf to} $m$}
     {
       \lIf{$i\in I$}{$\varphi_i\gets 1$ }\lElse{$\varphi_i\gets \xi_{\pi_i(1)}^{\pi_i(\ell_i)}$}
     }} }  

\caption{Worst-case backlog algorithm}
\label{algo:wcb}
\end{algorithm}

\begin{theorem}
  \label{th:wcb}
  Consider a tree network with $n$ servers offering rate-latency strict
  service curves $\beta_{R_j,T_j}$, and $m$ flows with leaky-bucket
  arrival curves $\gamma_{b_i,r_i}$. Let $I$ be a subset of flows
  crossing server $n$. Then there exists $(\rho_j)_{j\in\N_n}$ and
  $(\varphi_i)_{i\in \N_n}$ such that the worst-case backlog at server
  $n$ for flows in $I$ is
  \begin{equation}
    \label{eq:wcb}
    B = \sum_{j=1}^n\rho_j T_j + \sum_{i=1}^m \varphi_i b_i,
  \end{equation}
  where the
  coefficients $\rho_j$ and $\varphi_i$ depend only on $r_i$ and $R_j$
  and are computed by Algorithm~\ref{algo:wcb}. This algorithm runs in time $O(n^2+m)$.
\end{theorem}

If there is only one flow for each possible source/destination pair, then $m\leq n^2/2$ and the algorithm runs is $O(n^2)$. 

\begin{proof}[Sketch of the proof]
  The proof of the theorem is based on the construction of an admissible 
  trajectory ({\em i.e.} cumulative functions for each flow, at the
  input/output of each server in the path of this flow, that respect
  the input and output constraints given by the arrival and service
  curves) whose backlog at server $n$ is maximal for the flows in $I$
  (we call it a {\em worst-case trajectory}).

  Similar to the proof in~\cite{BN15}, the proof is in two steps. First, we show that there exists a
  worst-case trajectory that satisfy some properties. The second step
  is to construct a worst-case trajectory among the trajectories
  having those properties.

  {\em Properties of a worst-case trajectory:}  suppose
  that the worst case backlog is obtained at time $t_{n+1} =
  t_{n^\bu}$. There exists a
  worst-case trajectory that satisfy the following properties.
  \begin{enumerate}
  \item The
  service policy is SDF (shortest-to-destination-first).
\item For each server $j$, there is a unique backlogged period $[t_j,
  t_{j^\bu}]$, where the service offered is as small as possible. 
\item The arrival function of flow $f_i$ entering the system at server
  $j$ 
is maximal from $t_j$, the start of
  the backlogged period of server $j$ 
for all
  $t>t_j$ and 0 otherwise.
\item Data from the flows of interest in $I$ crossing server $j$ are
  instantaneously served at time $t_{j^{\bu}}$ and are all in server $n$ at time $t_{n+1}$. 
  \end{enumerate}

These properties are straightforward generalizations from~\cite{BN15} to trees.

\smallskip

  {\em Worst-case trajectory with the properties:} Once the set of
  trajectories has been restricted to the one satisfying the four
  properties above, the only optimization remaining is choosing the dates
  $t_j$. Indeed, if dates $t_j$ are fixed, the four properties above
  exactly determines the trajectory. Intuitively, the larger the
  backlog transmitted to the next server, the larger the backlog at
  server $n$. The maximization of the transmitted backlogs is done by
  a backward induction, from the root of the tree (server $n$) to the
  leaves that is detailed in Appendix~\ref{appA}. This
  optimization is then translated into Algorithm~\ref{algo:wcb}. 
\end{proof}

\paragraph{Worst-case delay}
The worst-case delay of a flow can be deduce from the worst-case backlog when $I$ is reduced to this flow. 

\begin{corollary}
  \label{cor:wcd}
  Suppose that flow $1$ crosses server $n$. Then the worst-case delay
  of flow $1$ starting at server $j$ and ending at server $n$ is
  $$\Delta = \frac{B-b_1}{r_1} + \frac{\xi_j^nb_1}{r_1},$$ where $B$ and $\xi_j^n$ are the worst-case backlog and coefficient obtained from Algorithm~\ref{algo:wcb} when $I = \{1\}$. 
\end{corollary}

\begin{proof}
  Suppose the flow of interest is $f_1$, with starting at server
$\pi_1(1) = j$ and ending at server $\pi_1(\ell_1) = n$. We are
interested in computing the worst-case delay of this flow. From~\cite{BN15}, the worst
case delay is obtained for the bit of data $b_1$. We can then
do the following modification: assume there are $m+1$ flows. Flows
$f_2$ to $f_m$ remain unchanged, flow $f_1$ has arrival curve
$t\mapsto b_1$ and flow $f_0$ has arrival curve $t\mapsto r_1 t$. The
flow on interest is now flow $f_0$. The worst-case backlog is obtained
at time $t_{n+1}$, and is $r_1(t_{n+1}-t_{j})$. It is
maximal when $t_{n+1}-t_{j}$ is, and then is the worst-case delay for
the first bit is data of flow $f_0$, which is equivalent to bit of
data $b_1$ of the original network.

Let $B$ be the worst-case backlog obtained with Algorithm~\ref{algo:wcb} with the original network. With the modified network, the backlogs becomes $B' = B-b_1 + \xi_j^n b_1$, as the new flow $f_1$ is not of interest (in the transformation, $\varphi_1$ changes from 1 to $\xi_j^n$), then the worst-case delay is $B'/r_1$, which corresponds to the desired result. 
\end{proof}

\paragraph{Application to sink-trees}
  Sink-trees are tree topologies where the destination of every flow
  is the root (node $n$). In this special case, each iteration of the
  external loop (lines 5-13) can be performed in constant time (there
  is only one test to perform). Moreover, the number of flows is at
  most the number of servers. As a consequence, the algorithm can be
  performed in $O(n)$. This type of topology has been studied in the
  context of {\em Sensor Network Calculus}. In \cite{BS15}, the
  authors give a close-form formula for the maximum backlog at the
  root and the end-to-end delay of a flow of interest. 

  Concerning the maximum backlog, this corresponds in our algorithm to
  the case where every flow is a flow of interest, so $\phi_i = 1$ for
  each flow $i$ and $\rho_j = r_j^*$. It is easy to see that the
  formula is the same as in~\cite[Theorem~14]{BS15}. 

 \begin{figure}[htbp]
   \centering
   \begin{picture}(0,0)%
\includegraphics{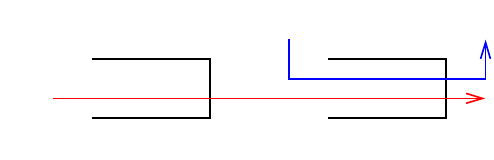}%
\end{picture}%
\setlength{\unitlength}{4144sp}%
\begingroup\makeatletter\ifx\SetFigFont\undefined%
\gdef\SetFigFont#1#2#3#4#5{%
  \reset@font\fontsize{#1}{#2pt}%
  \fontfamily{#3}\fontseries{#4}\fontshape{#5}%
  \selectfont}%
\fi\endgroup%
\begin{picture}(2254,718)(4621,-3671)
\put(5761,-3076){\makebox(0,0)[lb]{\smash{{\SetFigFont{9}{10.8}{\rmdefault}{\mddefault}{\updefault}{\color[rgb]{0,0,1}$\gamma_{b,r}$}%
}}}}
\put(6256,-3616){\makebox(0,0)[lb]{\smash{{\SetFigFont{9}{10.8}{\rmdefault}{\mddefault}{\updefault}{\color[rgb]{0,0,0}$\beta_{2R,T}$}%
}}}}
\put(5176,-3616){\makebox(0,0)[lb]{\smash{{\SetFigFont{9}{10.8}{\rmdefault}{\mddefault}{\updefault}{\color[rgb]{0,0,0}$\beta_{R,T}$}%
}}}}
\put(4636,-3346){\makebox(0,0)[lb]{\smash{{\SetFigFont{9}{10.8}{\rmdefault}{\mddefault}{\updefault}{\color[rgb]{1,0,0}$\gamma_{b,r}$}%
}}}}
\end{picture}%
     \caption{Example of sink-tree.}
     \label{fig:2-2}
   \end{figure}
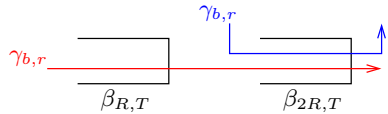

  Concerning the end-to-end delay of a flow, we show on a simple
  example that our approach leads to tighter delays. Consider the toy
  example with two servers in tandem described in Figure~\ref{fig:2-2}. 
  The worst-case delay bound from~\cite[Theorems~18 and 15]{BS15} is 
  $$D_1 = 2T + \frac{2b+rT}{R}.$$
  
  With our algorithm, we compute $\xi_{2}^2 = \frac{r}{2R-r}$ and
  $\xi_1^2 = \frac{r}{R}$, so the worst-case delay is, from Corollary~\ref{cor:wcd},
$$D_2 = 2T + \frac{b}{R} + \frac{b + rT}{2R-r}.$$
As $R>r$, $D_2 <D_1$. This is quite intuitive, as the cross-traffic
arrives at server 2, and then is served at rate $2R$. 

\paragraph{Arrival curve for the departure processes}
\begin{corollary}
  With the same notations as in Theorem~\ref{th:wcb}, the arrival
  curve of the departure functions from server $n$ for flows in $I$ is
  $\gamma_{B,\sum_{i\in I} r_i}$.
\end{corollary}

This is a direct consequence of Theorem~\ref{th:wcb} and of Lemma~\ref{lem:wcbI}.

\begin{lemma}
  \label{lem:wcbI}
  Consider a system and flows crossing that system being globally constrained by the arrival curve
  $\gamma_{b,r}$. If the maximal backlog of these flows in this
  system is less than $B$, then the arrival curve for the departure
  process of these flows is constrained by $\gamma_{B,r}$. 
\end{lemma}
\begin{proof}
Let $F^{(in)}$ be the sum of the arrival processes of the
  flows of interest and $F^{(out)}$ the sum of departure processes. Fix $s<t$, and transform $F^{(in)}$ from time $s$: for each flow of interest, the arrival process becomes maximal:
  a burst arrives at time $s$, and then data arrival at rate
  $r$. Call $F'$ this process. We have $F'(s) - F^{(out)}(s)\leq B$, by hypothesis, and $F^{(out)}(t) \leq F^{(in)}(t)
  \leq F'(s)+r(t-s)$. So $F^{(out)}(t)  - F^{(out)}(s) \leq F'(s) + r(t-s) - F^{(out)}(s) \leq B+r(t-s).$
\end{proof}

\section{Stability and performance bounds in cyclic networks}
\label{sec:stability}

In this section, we combine the results of the two previous
sections. We first restrict to the linear model, and in the last
paragraph, we show how those results could be extended to more general
cases.

\subsection{One-stage optimization problem}
We first investigate the one-stage optimization problem of
Equation~\eqref{1stage}. Given a network, several transformations are
possible, and we give here three of them. Due to the linear model, the
optimization problem boils down to a linear program: there exist a
non-negative matrix $M\in \R^{L,L}$, a non-negative column-vector
$N\in \R^L$, a non-negative line-vector $Q \in \R^L$ and a
non-negative constant $C$ such that performance $P$ can be compute as
\begin{equation}\label{1stage-lp}\begin{array}{|l|}
  \hline
  \text{Maximize }Q\bb + C \text{ such that }$${\mathbf{b}} \leq M{\bb} + N$$.\\
  \hline
\end{array}
\end{equation}

A stability condition is given by the following theorem:

 \begin{theorem}
    \label{th:TD}
    If the spectral radius of $M$ is strictly less than 1, then $\cN$ is stable.
\end{theorem}

As a consequence, depending on the decomposition we will obtain
different stability conditions.  In the following, we only explicit
the construction of $M$ and $N$. Vector $Q$ and constant $C$ can be
computed by similar methods.

\subsubsection{Server decomposition}

In the literature, the most usual decomposition is into
elementary servers ($\A^\r = \A$) and elementary flows (no
grouping). With our notation,
$S = \{(i,k)~|~i\in\N_m, 1\leq k\leq \ell_i\}$ and
$\pi_{i,k} = \langle \pi_i(k)\rangle$, and
$L = |S| = \sum_{i\in\N_n} \ell_i$. The decomposition of $S$ is into
singletons, and if $S^\ell = \{(i,k)\}$, we simply denote by
$\alpha_{i,k}$ the smallest arrival curve of $f_{i,k}$.

From classical results (see \cite[Sec. 6.3.2]{LT2001} for example), for all
$k<\ell_i$, 
\begin{equation}
\label{eq:sfa}
\alpha_{i,k+1} \leq \alpha_{i,k}\deconv (\beta_j - \sum_{s\in S_j\setminus \{(i,k)\}} \alpha_s)_+,
\end{equation}
and $\alpha_{i,1} = \alpha_i$, which  gives for leaky-bucket and rate-latency curves, and under stability assumption, that $\alpha_{i,k} =  \gamma_{b_{i,k},r_i}$ and from  Equation~\eqref{eq:sfa},
$$b_{i,k+1} \leq b_{i,k} + \frac{r_i}{R_j - \sum_{p\in \fl(j)\setminus\{i\}}r_p} (\sum_{s\in S_j\setminus \{(i,k)\}} \!\!\!\!\!b_s + R_jT_j).$$

Parameters $T_j$, $R_j$ and $r_i$ are fixed and 
$b_s$ are variables, this equation gives the coefficients of $M$ and $N$, that we denote $M_\SD$ and $N_{\SD}$ in the following. 

\subsubsection{Tree decomposition}

In this paragraph, we use the decomposition into a tree instead
of decomposing the network into elementary servers.

Suppose that arcs $\A^\r$ have been removed such that the remaining
network is a tree networks. In this case, as a tree has exactly $n-1$
arcs, so
$L = |S| \leq \sum_{i\in\N_n} \ell_i -n+1$.

Consider $a=(j_1,j_2)\in\A^\r$ and $f_{i,k}$ a flow such
that $j_1=\pi_{i,k}(\ell_{i,k})$ and $j_2=\pi_{i,k+1}(1)$.
An arrival curve for flow $f_{i,k+1}$ can be computed from the
others: from Lemma~\ref{lem:wcbI}, an arrival curve for flow $f_{i,k+1}$
is $\gamma_{b_{i,k+1},r_i}$ where $b_{i,k+1}$ is the maximum
backlog for flow $i$ at server $j_1$ computed with
Algorithm~\ref{algo:wcb}. As $b_{i,k+1}$ is linear in the bursts of
the other flows, there exists $(\varphi_s^{i,k+1})_{s\in S}$ and $(\rho_j^{i,k+1})_{j=1}^n$  such that 
$$b_{i,k+1} \leq \sum_{s\in S} \varphi_s^{i,k+1} b_s + \sum_{\{j|j\path j_1\}} \rho_j^{i,k+1}T_j,$$
where the exponent $i,k+1$ emphasizes the fact that the backlog 
computed is the burst parameter of
$\alpha_{i,k+1}$. 

As a consequence, with $M_{\TD}$ and $N_{\TD}$ playing the role of $M$
and $N$ above, we have $(M_\TD)_{s,s'} = \varphi_{s'}^{s}$ and
$(N_\TD)_s = \sum_{\{j|j\path j_1\}} \rho_j^{s}T_j$.

\begin{example}
  Consider the tree decomposition of Figure~\ref{fig:toy-unfold} is
  obtained. 

To find the other equations, we apply Algorithm~\ref{algo:wcb} and with the notations above, 
  \begin{equation}
    \label{eq:toy-round1}
    \left\{\begin{array}{l}
      b_{1,2} =  b_{1,1} + \varphi_{2,1}^{1,2} b_{2,1} + \varphi^{1,2}_{1,2} b_{1,2} + Cst_1\\
      b_{2,2} =  b_{2,1} + \varphi_{1,1}^{2,2} b_{1,1} + \varphi^{1,2}_{1,2} b_{1,2} + Cst_2\\
      b_{3,2} =  b_{3,2} + \varphi_{1,2}^{3,2}(b_{1,2} + b_{2,2} + b_{4,1}) + \rho_2^{3,2}T_2.
    \end{array}
    \right.
  \end{equation}
  Note that the expression of $b_{3,2}$ only depends on the behavior
  of server 2. Indeed, Algorithm~\ref{algo:wcb} only explores a server 
  and its descendants, but server 2 has none. Also, note that from
  Algorithm~\ref{algo:wcb}, two flows following the same path have the
  same linearity coefficient: $\varphi_{1,2}^{3,2} =
  \varphi_{2,2}^{3,2} = \varphi_{4,1}^{3,2}$.
\end{example}

As the backlog bounds computed with Algorithm~\ref{algo:wcb} are tight, the stability condition  with matrix $M_{\TD}$ is better than that with matrix $M_{\SD}$. 

\subsubsection{Arc grouping}
\label{sec:ag}
Despite the fact that Algorithm~\ref{algo:wcb} computes the worst-case
backlog bound for each flow, $M_\TD$ having spectral radius less than
one is only a sufficient condition for the network stability. Indeed,
the linear system is obtained by running Algorithm~\ref{algo:wcb}
independently $|S|$ times, but worst-case bounds for flows $s$ and $s'$
ending at the same server do not happen at the same time. Indeed,
consider two flows with respective arrival curve $\gamma_{b_1,r_1}$
and $\gamma_{b_2,r_2}$ crossing a server offering a strict service
curve $\beta_{R,T}$. Then the worst-case delay for flow 1 is $B_1 =
b_1 + \frac{r_1}{R-r_2}(b_2 + RT)$, for flow 2 is $B_2=b_2 +
\frac{r_2}{R-r_1}(b_1 + RT)$ and the worst-case backlog in the server
is $B = b_1 + b_2 + (r_1+r_2)T$. Obviously, $B < B_1+B_2$.

In this paragraph, our strategy is to group flows according to the removed arcs.  

Suppose that the network is stable and denote by $B_a$ the worst-case
backlog at arc $a=(j_1,j_2)\in \A^\r$, that is the maximal backlog at
server $j_1$ of flows having $\langle j_1,j_2\rangle$ as a sub-path. We denote $S_a
= \{(i,k)\in S~|~\pi_{i,k}(\ell_{i,k}) = j_1 \text{ and }
\pi_{i,k+1}(1) = j_2\}$ and $S'_a = \{(i,k+1)\in
S~|~\pi_{i,k}(\ell_{i,k}) = j_1 \text{ and } \pi_{i,k+1}(1) = j_2\}$.

With Algorithm~\ref{algo:wcb}, one can compute an upper of bound
$B_a$ for each $a\in\A^\r$: $\exists \varphi^a_{s}$ and $\rho_j^a$ such
that
$$
B_a \leq  \sum_{s\in S} \varphi_{s}^a b_{s}+  \sum_{\{j|j\path j_1\}} \rho^a_j T_j.
$$

The next step is to refine this equation so that $(B_{a})_{a\in\A^\r}$ appear in
the right-hand term instead of $(b_{s})_{s\in S}$. We know from the proof of
Theorem~\ref{th:wcb} that the worst-case backlog is maximized when the
cross-traffic is maximal.  Consider arc $a' =
(j'_1,j'_2)\in\A^\r$. For all $s\in S'_{a'}$, the arrivals of $f_s$
will all be maximized from time $t_{j'_2}$. At this time, the backlog
in server $j'_1$ is at most $B_{a'}$ and the backlog of each flow
transmitted to server $j'_2$ is $x_s$ with $\sum_{s\in S'_{a'}} x_s
\leq B_{a'}$. From time $t_{j'_2}$ on, data of flow $f_s$
necessarily arrives at rate $r_s$: if it could arrive faster, the
backlog would not have been maximized.

As a consequence, for all $a'=(j'_1,j'_2)\in\A^\r$, if $B_{a'}$ is the
worst-case backlog at server $j'_s$, for flows $f_s$, $s\in S'_{a'}$,
there exists $(x_s)_{s\in S'_a}$ such that $\sum_{s\in S'_{a'}}x_s
\leq B_{a'}$ and
\begin{align*}
  B_a  \leq & \sum_{s\in S} \varphi_{s}^a x_{s}+  \sum_{\{j|j\path j_1\}} \rho^a_j T_j \\
   \leq & \sum_{a'\in \A'} \Big[\Big(\max_{s\in S'_{a'}}\varphi_{s}^a\Big) \Big(\sum_{s\in S_a'}x_{s}\Big) \Big]+ \sum_{i=1}^m \varphi_{(i,1)}^a b_i+ \sum_{\{j|j \path j_1\}} \rho^a_j T_j \\ 
   \leq & \sum_{a'\in \A'}\Big(\max_{s\in S'_{a'}}\varphi_{s}^a\Big)B_{a'} +  \sum_{i=1}^m \varphi_{(i,1)}^a b_i+\sum_{\{j|j \path j_1\}} \rho^a_j T_j.
\end{align*}
As a consequence, with $M_{\TDA}$ and $N_{\TDA}$ playing the role of $M$
and $N$ above, we have $(M_{\TDA})_{a,a'} = \max_{s\in S_{a'}}\varphi_{s}^a$ and
$(N_{\TDA})_{a} = \sum_{i=1}^m \varphi_{(i,1)}^a b_i+\sum_{\{j|j \path j_1\}} \rho^a_j T_j$.

\begin{example}
  We compute backlog bounds for arcs $a_1=(4,2)$ and $a_2=(2,1)$ and obtain. 
\begin{equation*}
    \label{eq:toy-round2}
    \left\{\begin{array}{ll}
      B_{a_1} & \leq  N_{\TDA}^{a_1} +  (\varphi_{1,2}^{a_1} \lor \varphi_{2,2}^{a_1})B_{a_1} + \varphi_{3,2}B_{a_2} \\
      B_{a_2} &\leq  N_{\TDA}^{a_2} + \varphi_{1,2}^{a_2}B_{a_1}.
    \end{array}
    \right.
  \end{equation*}
  A sufficient condition for the stability is then given by $(\varphi_{1,2}^{a_1} \lor \varphi_{2,2}^{a_1})+ \varphi_{3,2} \varphi_{1,2}^{a_2}< 1$. 
\end{example}

It is not possible to compare the stability bound with $M^\TDA$ with
$M^\TD$ or $M^\SD$: there are examples where the stability bound will
be better, for the unidirectional ring for example, and some examples
where it will be worse, like for the bidirectional ring. In the next
section, we present those two examples that illustrate the advantages
and limits of this latter approach.

\subsection{Examples}
\subsubsection{Stability of the unidirectional ring} 
\label{sec:1ring}
Consider a ring with $n$
nodes. Its induced graph is $\cG$ with
$\A = \{(i,i+1), i\leq n-1\}\cup \{(n,1)\}$. The transformation into a
tree gives a tandem networks by removing arc $(n,1)$. Flows are decomposed in either one flow or two flows. Grouping flows
that cross this arc enables to show the stability of the
unidirectional ring. 

\begin{theorem}
  \label{th:ring}
  The unidirectional ring is stable under local stability condition.
\end{theorem}

\begin{proof}
  We consider matrix $M_{\TDA}$ and take $\A^\r=\{(n,1)\}$. With $I$ the set of flows that circulate through arc $(n,1)$, 
  $S_{(n,1)} = \{(i,2)~| i\in I\}$.  When computing the
  worst-case backlog at arc $a$, the flows of interests are flows
  $f_{(i,1)}$ for $i \in I$ and
  $$B_a \leq
  \max_{s\in S_a}{\varphi_s^a} B_a + C,$$
  where $C$ is a constant there is no need to explicit in this proof. So it remains to
  show that for all $s\in S_a$, $\varphi^a_s<1$. As $(i,2)\in S_a$ is
  not a flow of interest, $\varphi_{(i,2)} = \xi_{1}^{\pi_i(\ell_i)}$. Observe from Algorithm~\ref{algo:wcb} how $\xi_j^{\ell}$ are
  computed: because of the local stability, $R_n >r_n^n + r_n^*$, so
  $\xi_n^n <1$. Now assume that $\xi_{j^{\bu}}^k <1$ (lines
  7-11). Either $\xi_j^k = \xi_{j^{\bu}}^k <1$, or
  $\xi_j^\ell =(r^*_j + \sum_{\ell'\in k^{\bu} \path n}
  \xi_{j^\bu}^{\ell'} r_j^{\ell'})/ (R_j-\sum_{\ell\in j \path
    k}r_j^{\ell'}) \leq (r^*_j + \sum_{\ell'\in k^{\bu} \path n}
  r_j^{\ell'})/ (R_j-\sum_{\ell\in j \path k}r_j^{\ell'})<1$,
  as from local stability condition.
  As a consequence, for all $j$ and $\ell$, $\xi_j^\ell <1$ and
  $\max_{s\in S_a}{\varphi_s^a}<1$, and 
  $B_a \leq C(1-\max_{s\in S_a}{\varphi_s^a})^{-1}$,  ensuring
  the stability of the network.
\end{proof}

This result has already been proved under stronger assumptions:
in~\cite{TG1996} when servers are constant-rate servers and
in~\cite{LT2001} when servers have a maximal service rate. Our method
is not specific to the ring topology, so we can hope to improve the
stability conditions for more general topologies. 

\subsubsection{The bi-directional ring}
\label{sec:bidir}
An example where grouping according to the arcs is not be efficient is
the bi-directional ring with $n$ servers.  Suppose the network
is crossed by $2n$ flows of length $n$: $\pi_1 = \langle
1,2,\ldots,n\rangle$, $\pi_{n+1} = \langle n,n-1,\ldots,1\rangle$,
$\pi_i = \langle i,i+1,\ldots,n,1\ldots,i-1\rangle$ and $\pi_{n+i}
=\langle i,i-1,\ldots,1,n\ldots,i+1\rangle$ for $i = 2,\ldots,n$.  The
tree decomposition is obtained by keeping arcs $\{(i,i+1), i\leq
n-1\}$ and the path obtained after the decomposition are the one
obtained for the unidirectional ring for $f_1,\ldots,f_n$, and flows
of length 1 for the other paths.

With this decomposition, we can never ensure stability: let us look at
the coefficient $\varphi_{a'}^a$ that are computed. Consider arc
$a=(2,1)$ for example. Among the flows of interest are the flows
$f_{(i,k)}^a$ of path $\langle 2\rangle$, with $k\neq 1$, so
$\varphi_{(i,k)}^a=1$. This means that $(M_{\TDA})_{(2,1),(3,2)}=1$,
and the similarly, $(M_{\TDA})_{(j,j-1),(j+1,j)}=1$ and
$(M_{\TDA})_{(1,n),(2,1)}=1$. There is a cycle of coefficients 1 in
the matrix: the spectral radius of $M_\TDA$ is at least 1.

More generally, grouping according to the arcs will never ensure the
stability if in matrix $M_\TDA$ it is possible to find a cycle with
weights one on all its arcs. As a consequence, intermediate solution
between no grouping of flows and grouping among the arcs might lead to
better solutions. For example, in the case of the bi-directional ring,
a better solution would be to group flows for the removed arc $(n,1)$ only,
and not group the other flows.

\subsection{Two-stage optimization problem}

We have seen in through the examples of the previous paragraph that
different stability conditions and performance bounds can be found,
depending on how the network is decomposed. Following the approach of
Equation~\eqref{2stage}, it is possible to combine the optimization
problems: 

$$
\begin{array}{|rl|}
 \hline
  \text{ Maximize }&Q\bb' + C \\ \text{ such that }&
  \mathbf{b}'\leq \mathbf{b}, \sum_{s\in S_a} b'_s \leq B_a \\ 
  & {\mathbf{b}} \leq M_\TD{\bb} + N_\TD,~\mathbf{B} \leq M_\TDA \mathbf{B} + N_\TDA.\\
  \hline 
\end{array}
$$
This formulation slightly differs from the one in Equation~\eqref{2stage}:
constraint `$b'_s \leq B_a\quad \forall s\in S_a$'' has been replaced
by ``$\sum_{s\in S_a} \mathbf{b}'_s \leq B_a$''. Indeed, by a
reasoning similar to that of Paragraph~\ref{sec:ag}, we can fix
$a = (i,j)\in \A^r$. If the worst-case backlog at server $i$ for
flows crossing $a$ is $B_a$ this means that when this worst-case
happens, there is no data of these flows in the rest of the network
(which would deny the maximality of $B_a$).  Consider a flow $f_s$,
$s\in S_a$. Its amount of data in arc $a$ is $x_s$, and its arrival
rate $r_s$. Data cannot arrive faster than $r_s$, so from the time of
worst-case backlog, flow $f_s$ is $\gamma_{x_s,r_s}$-constrained.

\subsection{General arrival and service curves}
Beyond the linear model, it should be possible to obtain tighter
bounds by using more general arrival and service curves. For example,
stair-case functions, or piece-wise linear arrival curves and service
curves.

A first remark is that the stability conditions given here only depend
on the arrival and service rates, then, in the case it is possible to
refine these results to more general curves (as it is for the SD
method), no better stability condition can be inferred. Indeed, a
general curve can usually be lower and upper-bounded by two token-bucket
curves, inducing a lower and an upper-bound of the network by two
linear models with the same stability condition.

The second remark is that it would still be possible to improve the
performance bounds. To our knowledge, there is no evidence in the
literature that the equation $\ba = H(\ba)$ has a unique
fix-point in the general case. One safe solution is to compute the
greatest fix-point by iterations methods. The first step of this
approach is to find an upper bound of that greatest fix-point. This
can be done by computing the fix point of that equation in the linear
model (by bounding the arrival and service curves by linear curves),
and the second step is to iterate from that point for refining the
performance bounds. At each iteration, the performance bound obtained
is an upper bound of the performance of the network.

\section{Comparison and numerical experiments}
\label{sec:numerical}
The different approaches have been implemented in Python and run on a
basic laptop. We will not comment on the computational time as all
those algorithms are polynomial, and the number of constraints of our
linear programs are linear in the size of the networks.

We call $\SD$ the server decomposition method,
$\TD$ the tree decomposition method, $\TDA$ the arc grouping method and $\TDC$ 2-stage method.

We compare those methods on three examples: the unidirectional ring,
the bidirectional ring and a 3-ring network.

In the experiments we assume that the utilization rate of the
network is $U = \min_{j=1}^n U_j$.
flows have uniform parameters:
$b_i = 1kb$, $r_i = 1kb.s^{-1}$ for all $i\in\N_m$, and $T_j = 10ms$ for all
$j\in\N_n$. Only the service rate will vary in function of an
utilization rate: the utilization rate of server $j$ is $U_j =
\frac{\sum_{i\in\fl(j)} r_i}{R_j}$, and

\subsection{The unidirectional ring}
We now consider the example of the unidirectional ring described in Section~\ref{sec:1ring}
with $n=10$.

Figure~\ref{fig:num-ring-uni} shows the backlog guarantee at server
$n$ of flow $1$ for uniform traffic: servers all have the same service
rate $R = 10/U kbs^{-1}$.  The stability conditions are $U<0.18$ for
$\SD$ and $U<0.62$ for $\SD$, so our methods greatly improves the
stability region.  We notice that $\TDA$ is better than $\TD$, so
$\TDA$ and $\TDC$ compute the same bounds.

\begin{figure}[htbp]
	\centering
\begin{tikzpicture}
  \begin{semilogyaxis}[height = 7cm, width = 12cm, xlabel={Utilization rate}, ylabel = {Backlog upper bound},legend entries={$\SD$, $\TD$, $\TDA$}, legend style = {at = {(0.03,0.97)}, anchor = north west},ymin = 1,ymax = 1000,xmin=0, xmax = 1]
    \addplot+[mark = none,dashed] table[x index = 0,y index= 1] {cycleunif.txt};
    \addplot+[mark = none] table[x index = 0,y index= 2] {cycleunif.txt};
    \addplot+[mark = none, thick] table[x index = 0,y index= 3] {cycleunif.txt};
   \end{semilogyaxis}
\end{tikzpicture}
 \caption{Backlog bound for the unidirectional ring and
    uniform servers.}
  \label{fig:num-ring-uni}
\end{figure}
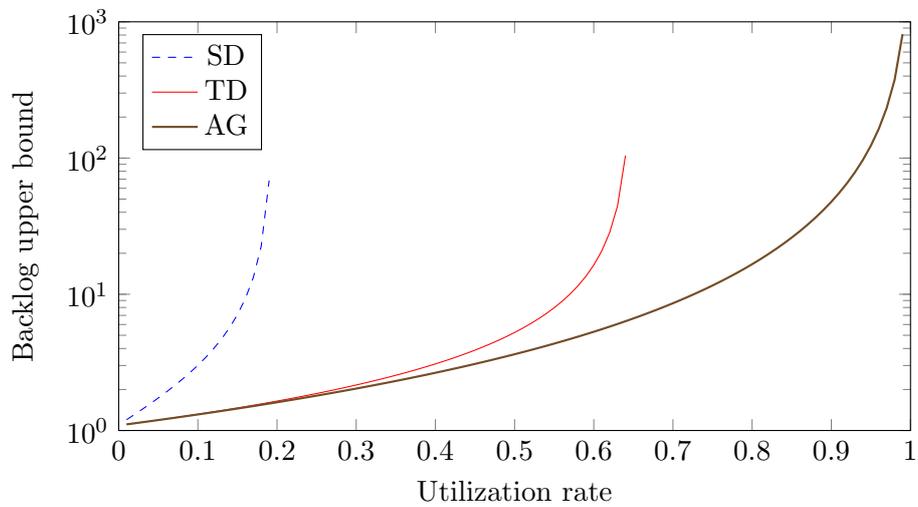

Fig.~\ref{fig:stab-ring} (left) shows the ratio between the stability bounds with $\SD$ and with $\TD$ as the number of servers increases on the ring. The ratio grows linearly. 
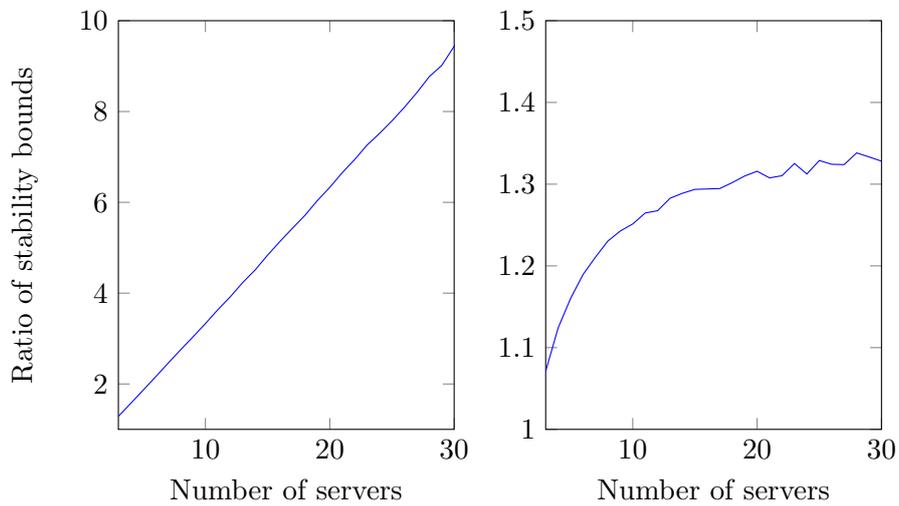
\begin{figure}
	\centering
\begin{tikzpicture}
  \begin{axis}[enlarge x limits = false,height = 7cm, width = 6cm, xlabel={Number of servers}, ylabel = {Ratio of stability bounds},ymin = 1,ymax = 10,xmin=3, xmax = 30]
    \addplot+[mark = none] table[x index = 0,y index= 3] {stabilitybis.txt};
   \end{axis}
 \end{tikzpicture}
 \begin{tikzpicture}
\begin{axis}[height = 7cm, width = 6cm, xlabel={Number of servers}, ymin = 1,ymax = 1.5,xmin=3, xmax = 30]
   \addplot+[mark = none] table[x index = 0,y index= 3] {stabilityter.txt};
   \end{axis}
\end{tikzpicture}
 \caption{Ratio of stability bounds on the ring with the number of servers grows. Left: unidirectional ring; right: bidirectional ring.}
  \label{fig:stab-ring}
\end{figure}

Fig~\ref{fig:num-ring-hete} shows the  backlog guarantee of
flow 1 at server $n$ when the servers have different service rates:
every service rate is $R = 20/U kbs^{-1}$, except $R_9 = R_{10} =10/U
kbs^{-1}$. 

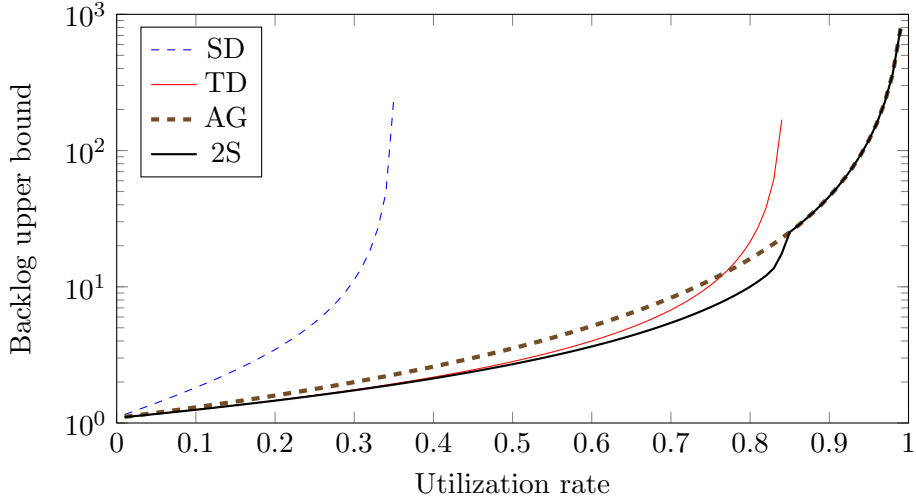
\begin{figure}[htbp]
\centering	
\begin{tikzpicture}
  \begin{semilogyaxis}[ height = 7cm, width = 12cm, xlabel={Utilization rate}, ylabel = {Backlog upper bound},legend entries={$\SD$, $\TD$, $\TDA$,$\TDC$}, legend style = {at = {(0.03,0.97)}, anchor = north west},ymin = 1,ymax = 1000,xmin=0, xmax = 1]
    \addplot+[mark = none,dashed] table[x index = 0,y index= 1] {cyclehete.txt};
    \addplot+[mark = none] table[x index = 0,y index= 2] {cyclehete.txt};
 \addplot+[mark = none, dashed, ultra thick] table[x index = 0,y index= 3] {cyclehete.txt};
    \addplot+[mark = none, thick] table[x index = 0,y index= 4] {cyclehete.txt};
   \end{semilogyaxis}
 \end{tikzpicture}
 \caption{Backlog bound for the unidirectional ring with heterogeneous servers.}
  \label{fig:num-ring-hete}
\end{figure}

In this case, $\TD$ is better than $\TDA$ for $U<0.68$, and around the
stability limit given by $\TD$, we observe $\TDC$ increases faster as
the backlog computed with $\TD$ grows to infinity. Then $\TDC$ and
$\TDA$ compute of course the same bound.

\subsection{The bidirectional ring}
We now consider the example of the bidirectional ring with $n=10$ as described in Paragraph~\ref{sec:bidir}. 

\begin{figure}[htbp]
	\centering
\begin{tikzpicture}
  \begin{semilogyaxis}[height = 7cm, width = 12cm, xlabel={Utilization rate}, ylabel = {Backlog upper bound},legend entries={$\SD$, $\TD$, group.,2stage}, legend style = {at = {(0.03,0.97)}, anchor = north west},ymin = 1,ymax = 1000,xmin=0, xmax = 0.25]
    \addplot+[mark = none,dashed] table[x index = 0,y index= 1] {bicycle.txt};
    \addplot+[mark = none] table[x index = 0,y index= 2] {bicycle.txt};
   \end{semilogyaxis}
 \end{tikzpicture}
  \caption{Backlog bounds for the bidirectional ring.}
  \label{fig:num-biring-uni}
\end{figure}
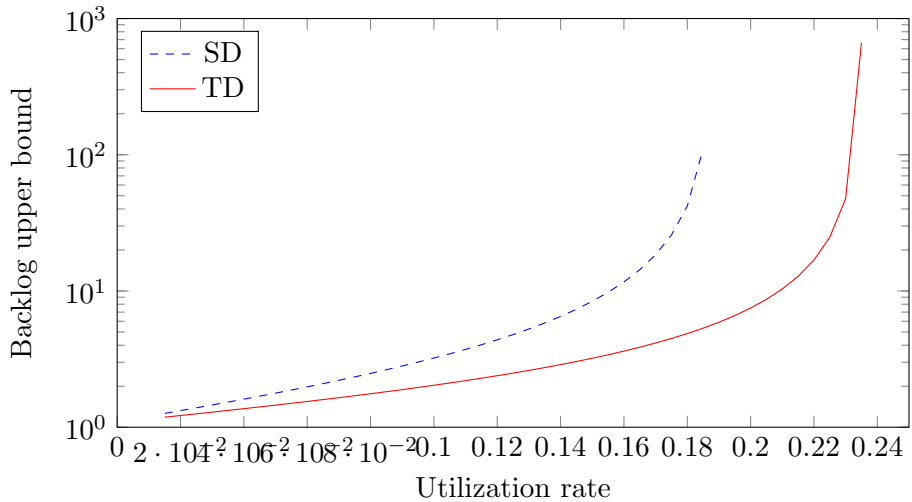

Figure~\ref{fig:num-biring-uni} shows the worst-case backlog bound of
flow 1 at server $n$ computed by the elementary decomposition and tree
transformation.  As expected, the stability condition with $\TD$
($U<0.24$) is improved from $\SD$ ($U<0.19$). The improvement is
approximately $25\%$. Fig.~\ref{fig:stab-ring} (right) shows the
improvement ratio when the number of servers grows. In this case, the
improvement seems logarithmic, and for $n=30$, it is approximately
$33\%$. $TD$ method suffers from having half the flows decomposed in
flows of length 1.

\subsection{A three-ring example}
The bidirectional cycle is not realistic, as in many network are
full-duplexed, but there might be several cycles in
network. Fig.~\ref{fig:tricycle} shows an example of a network
composed of three cycles, and flows circulate along one of the three
cycles. Three servers (those depicted) are common to two cycles. 

\begin{figure}[htbp]
  \centering
  \includegraphics{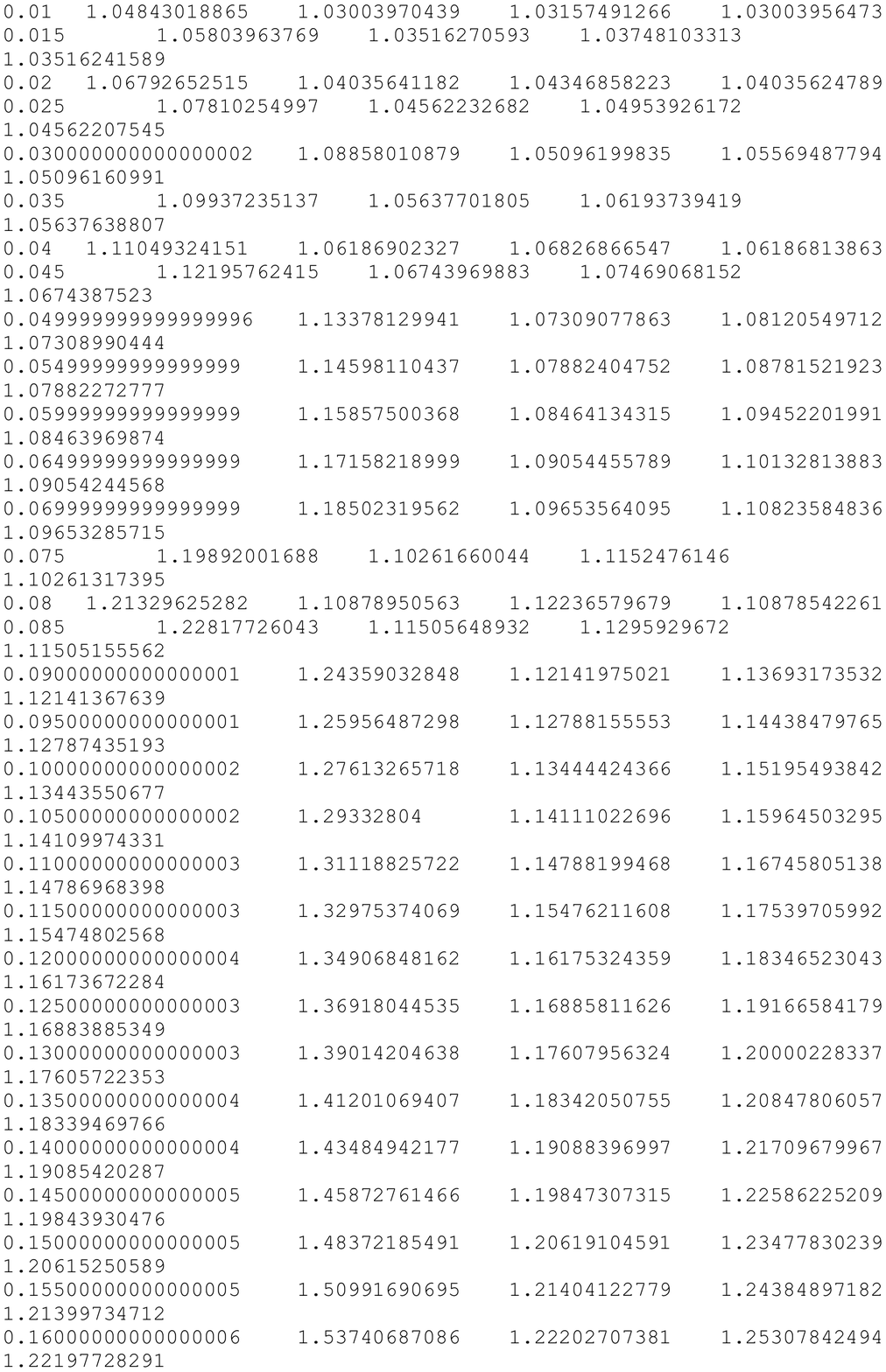}
  \caption{Network composed of three rings.}
  \label{fig:tricycle}
\end{figure}

Fig.~\ref{fig:num-tricycle} shows the backlog of a flow when each
cycle is made of 10 servers, and flows have length 10, except for one
cycle, where we use shorter flows to avoid the problem presented in
Paragraph~\ref{sec:bidir}. We can observe that the stability region
more than doubles from the $\SD$ ($U<0.34$) to the $\TD$ method
($U<0.73$). The improvement using grouping is smaller ($U<0.77$), but
still sensible.

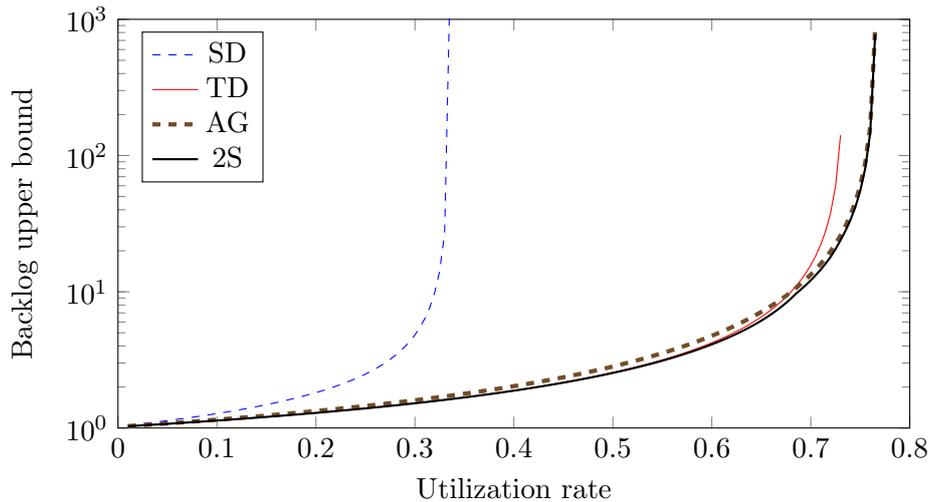
\begin{figure}[htbp]
	\centering
\begin{tikzpicture}
  \begin{semilogyaxis}[ height = 7cm, width = 12cm, xlabel={Utilization rate}, ylabel = {Backlog upper bound},legend entries={$\SD$, $\TD$, $\TDA$,$\TDC$}, legend style = {at = {(0.03,0.97)}, anchor = north west},ymin = 1,ymax = 1000,xmin=0, xmax = 0.8]
    \addplot+[mark = none,dashed] table[x index = 0,y index= 1] {tricycle.txt};
    \addplot+[mark = none] table[x index = 0,y index= 2] {tricycle.txt};
 \addplot+[mark = none, dashed, ultra thick] table[x index = 0,y index= 3] {tricycle.txt};
    \addplot+[mark = none, thick] table[x index = 0,y index= 4] {tricycle.txt};
   \end{semilogyaxis}
 \end{tikzpicture}
 \caption{Backlog bound for the three-ring example.}
  \label{fig:num-tricycle}
\end{figure}

\section{Conclusion}
In this article, the recent results from feed-forwards networks can be
adapted to improve the performance guarantees and stability conditions
of networks with general topology.

Many problems remain open and directions to investigate: finding the
transformation of the network that would lead to better guarantees,
adapt recent results for example~\cite{BS2016} that can be applied for
general arrival and service curves. Future works will also include the
adaptation to service policies, like the FIFO of static priority
policies.

More generally, the stability problem remains open.



\begin{thebibliography}{10}

\bibitem{ABS2004}
Carme {\`{A}}lvarez, Maria~J. Blesa, and Maria~J. Serna.
\newblock A characterization of universal stability in the adversarial queuing
  model.
\newblock {\em {SIAM} J. Comput.}, 34(1):41--66, 2004.

\bibitem{AMFLRU16}
A.~Amari, Ahlem Mifdaoui, Fabrice Frances, J{\'e}rome Lacan, David Rambaud, and
  Loic Urbain.
\newblock {AeroRing: Avionics Full Duplex Ethernet Ring with High Availability
  and QoS Management}.
\newblock In {\em ERTS}, 2016.

\bibitem{And2000}
M.~Andrews.
\newblock Instability of fifo in session-oriented networks.
\newblock In {\em Proceedings of SODA'00}, 2000.

\bibitem{And2007}
M.~Andrews.
\newblock Instability of {FIFO} in the permanent sessions model at arbitrarily
  small network loads.
\newblock In {\em Proceedings of SODA'07}, 2007.

\bibitem{AAFLLK2001}
Matthew Andrews, Baruch Awerbuch, Antonio Fern{\'{a}}ndez, Frank~Thomson
  Leighton, Zhiyong Liu, and Jon~M. Kleinberg.
\newblock Universal-stability results and performance bounds for greedy
  contention-resolution protocols.
\newblock {\em J. {ACM}}, 48(1):39--69, 2001.

\bibitem{BS15}
Steffen Bondorf and Jens~B. Schmitt.
\newblock Boosting sensor network calculus by thoroughly bounding
  cross-traffic.
\newblock In {\em Proceedings of INFOCOM 2015}, 2015.

\bibitem{BS2016}
Steffen Bondorf and Jens~B. Schmitt.
\newblock Improving cross-traffic bounds in feed-forward networks - there is a
  job for everyone.
\newblock In {\em {MMB} {\&} {DFT}}, pages 9--24, 2016.

\bibitem{BKRSW2001}
Allan Borodin, Jon~M. Kleinberg, Prabhakar Raghavan, Madhu Sudan, and David~P.
  Williamson.
\newblock Adversarial queuing theory.
\newblock {\em J. {ACM}}, 48(1):13--38, 2001.

\bibitem{BJT2010b}
A.~Bouillard, L.~Jouhet, and \'{E}. Thierry.
\newblock Comparison of different classes of service curves in network
  calculus.
\newblock In {\em WODES}, pages 316--321, 2010.

\bibitem{BJT2010a}
A.~Bouillard, L.~Jouhet, and \'{E}. Thierry.
\newblock Tight {P}erformance {B}ounds in the {W}orst {C}ase {A}nalysis of
  {F}eed {F}orward {N}etworks.
\newblock In {\em INFOCOM'10}, 2010.

\bibitem{BN15}
Anne Bouillard and Thomas Nowak.
\newblock Fast symbolic computation of the worst-case delay in tandem networks
  and applications.
\newblock {\em Perform. Eval.}, 91:270--285, 2015.

\bibitem{BNOT2010}
M.~Boyer, N.~Navet, X.~Olive, and \'{E}. Thierry.
\newblock The {PEGASE} project: precise and scalable temporal analysis for
  aerospace communication systems with network calculus.
\newblock In {\em ISOLA'10}, 2010.

\bibitem{Chang2000}
C.-S. Chang.
\newblock {\em Performance Guarantees in Communication Networks}.
\newblock TNCS, Springer-Verlag, 2000.

\bibitem{Cruz1991b}
R.L. Cruz.
\newblock A calculus for network delay, part {II}: Network analysis.
\newblock {\em IEEE Transactions on Information Theory}, 37(1):132--141, 1991.

\bibitem{Cruz1995}
R.L. Cruz.
\newblock Quality of service guarantees in virtual circuit switched networks.
\newblock {\em IEEE Journal on selected areas in communication}, 13:1048--1056,
  1995.

\bibitem{GJ1979}
M.R. Garey and D.S. Johnson.
\newblock {\em Computers and Intractability: A Guide to the Theory of
  NP-Completeness}.
\newblock W. H. Freeman, 1979.

\bibitem{LT2001}
J.-Y. Le~Boudec and P.~Thiran.
\newblock {\em Network Calculus: A Theory of Deterministic Queuing Systems for
  the Internet}, volume LNCS 2050.
\newblock Springer-Verlag, 2001.
\newblock revised version 4, May 10, 2004.

\bibitem{LPR2004}
Zvi Lotker, Boaz Patt{-}Shamir, and Adi Ros{\'{e}}n.
\newblock New stability results for adversarial queuing.
\newblock {\em {SIAM} J. Comput.}, 33(2):286--303, 2004.

\bibitem{MR2006}
J.~M. McManus and K.~W. Ross.
\newblock Video-on-demand over {ATM}: Constant-rate transmission and transport.
\newblock {\em IEEE J.Sel. A. Commun.}, 14(6):1087--1098, September 2006.

\bibitem{PSKL2004}
Francesco~De Pellegrini, David Starobinski, Mark~G. Karpovsky, and Lev~B.
  Levitin.
\newblock Scalable cycle-breaking algorithms for gigabit ethernet backbones.
\newblock In {\em Proceedings {IEEE} {INFOCOM}}, 2004.

\bibitem{RL2008}
G.~Rizzo and J.-Y.~Le Boudec.
\newblock Stability and delay bounds in heterogeneous networks of aggregate
  schedulers.
\newblock In {\em Proceedings of INFOCOM'2008}, 2008.

\bibitem{SKZ2002}
David Starobinski, Mark~G. Karpovsky, and Lev Zakrevski.
\newblock Application of network calculus to general topologies using
  turn-prohibition.
\newblock In {\em Proceedings {IEEE} {INFOCOM}}, 2002.

\bibitem{TG1996}
Leandros Tassiulas and Leonidas Georgiadis.
\newblock Any work-conserving policy stabilizes the ring with spatial re-use.
\newblock {\em {IEEE/ACM} Trans. Netw.}, 4(2):205--208, 1996.

\end{thebibliography}

\appendix

\section{Proof of Theorem~\ref{th:wcb}}
\label{appA}
To avoid introducing too many notations, we first prove the result for
tandem networks ($\A = \{(i,i+1)~|~i\in\N_{n-1}\}$). We then explain how
to adapt it to sink-trees.  

We prove the theorem by a backward induction on the servers.  Let us
denote by $B_j(x_j,\ldots,x_n,x^*)$ the backlog at time $t_{n+1}$ in
server $n$ when the backlog transmitted at time $t_{j}$ by server
$j-1$ to server $j$ is $x_k$ for the flows served by server $j-1$ and
that end at server $k$ ($k\geq j$) and $x^*$ for the flows of interest
crossing server $j-1$.

We will use the additional notations
\begin{itemize}
\item $b_j^k = \sum_{\{i\notin I,~\pi_i = j\path k\}} b_i$ is the size of the
  burst arriving at server $j$ at time $t_{j}$ belonging to flows
  starting at server $j$ and ending at server $k$;
\item $b^*_j = \sum_{\{i\in I|~\pi_i(1)=j\}}b_i$ is the burst of the
  flows of interests starting at server $j$.
\end{itemize}

If we are able to compute $B_j$ for all $j$, the worst-case backlog
is $B=B_1(0,\dots,0)$. We will show by induction that:
\begin{enumerate}
\item[({\bf A})] $B_j$ is linear in the $x_k$, $T_k$, $b^k_{\ell}$,
  $b^*_{\ell}$, $k\geq \ell\geq j$ and in $x^*$.  More precisely we
  can write
$$B_j((x_{k})_{k\succeq j},x^*) = C_j + x^* + \sum_{k=j}^n \xi_j^k x_k,$$
where $\xi_j^k$ only depends on the $R_k$'s and $r_i^{(*)}$'s, and $C_j$ is a
polynomial of degree 1 in   $T_k$, $b^k_{\ell}$,
  $b^*_{\ell}$, $k\geq \ell\geq j$ (with coefficients
depending on the $R_k$'s and $r_i^{(*)}$'s only).
\item[({\bf B})] $\xi_j^k \leq \xi_j^{k+1}$.
\end{enumerate}

This inequality ({\bf B}) is quite intuitive: the coefficient
$\xi_j^k$ roughly corresponds to quantity of data produced by a flow
starting at $j$ and ending at $k$ the rate grows when the length of
the path grows, as there is more chance to meet a slower server.

\subsection{Initialization - computation of $B_n$}

$B_n$ only depends on the burst that is transmitted at time $t_{n+1}$,
which we note $x_n^n$ for the flows that are not of interest and $x^*$
for the flows of interest. The worst-case backlog for the flows of
interests is obtained when all data from the other flows have been
served and none of the flows of interests:
\begin{align*}
B_n(x_n^n, x^*) =& b^*_n + x^* +  r^*_n (T_n+ \frac{x_n^n + b^n_n + r^n_nT_n}{R_n-r^n_n})\\
=& b^*_n + x^* +  r^*_n T_n + \xi_n^n Q_n^n\\
=& C_n + x^* + \xi_n^n Q_n^n, 
\end{align*}
with $Q_n^n = x_n^n + b^n_n + r^n_nT_n$, $\xi_n^n = \frac{ r^*_n}{R_n-r^n_n}$ and $C_n = b^*_n  + r^*_n T_n$. 
\subsection{Inductive step  - computation of $B_j$ from $B_{j+1}$}
Suppose that $B_{j+1}(x_{j+1},\dots,x_n, x^*) = C_{j+1} + x^* + \sum_{k=j}^n \xi_j^k x_k$. 

$B^k_{j}(x_j^j,\dots,x_j^n, x^*)$ is computed the following way: it
takes time
$\delta$ to serve flows ending at servers $j,\ldots,k$, and data from any
other flow is instantaneously transmitted to server $j+1$. This quantity of data is
$x_{j+1}^{\ell} = b_j^\ell + x_j^\ell + r_j^\ell \delta$ for $\ell >
k$ where $\delta$ satisfies
\[
\sum_{\ell=j}^k (x_j^\ell+b_j^\ell + r_j^\ell \delta) = R_j(\delta -
T_j)_+,
\]
{\em i.e.},
$\delta = T_j+ {\sum_{\ell=j}^k Q_j^\ell }/\big({R_j - \sum_{\ell=j}^k
r_j^\ell}\big)$ with $Q_j^k = b_j^k + x_j^k + r_j^k T_j$.
For $\ell > k$, the amount of data transmitted to server $j+1$ by flows (not of interest) ending at server $k\leq j$ is 
\[
x_{j+1}^{\ell} = Q_j^\ell + r_j^\ell \frac{\sum_{\ell=j}^k Q_j^\ell }{R_j -
\sum_{\ell=j}^k r_j^\ell}.
\]

Then the backlog at server $n$ can be expressed from $B_{j+1}$:
\begin{align*}
B_j^k(x_j^j,\dots,x_j^n, x^*) =& B_{j+1}(0\ldots,0,x_{j+1}^{k+1},\ldots,x_{j+1}^{n},b^*_j + x^* +  r^*_j \delta).\\
=& C_{j+1} + b^*_j + x^* +r^*_j\left(T_j+ \frac{\sum_{\ell=j}^k Q_j^\ell }{R_j - \sum_{\ell=j}^k
r_j^\ell}\right)   + \sum_{\ell>k} \xi_{j+1}^{\ell} \left(Q_j^\ell + r_j^\ell \frac{\sum_{\ell=j}^k Q_j^\ell }{R_j -
\sum_{\ell=j}^k r_j^\ell}\right)\\
=&C_j +x^* +     r^*_j \frac{\sum_{\ell=j}^k Q_j^\ell }{R_j - \sum_{\ell=j}^k
r_j^\ell}  + \sum_{\ell>k} \xi_{j+1}^{\ell} \left(Q_j^\ell + r_j^\ell \frac{\sum_{\ell=j}^k Q_j^\ell }{R_j -
\sum_{\ell=j}^k r_j^\ell}\right)\\
=& C_j +x^* + \sum_{\ell=j}^k \left(\frac{ r^*_j+\sum_{\ell>k}\xi_{j+1}^{\ell}r_j^\ell }{R_j -
\sum_{\ell=j}^k r_j^\ell}\right) Q_j^\ell  + \sum_{\ell>k}\xi_{j+1}^{\ell} Q_{j}^\ell
\end{align*}
with $C_{j} = C_{j+1} +b^*_j + r^*_jT_j$.

Note the other scenarios should have been taken into account, when part of the flows ending at server $k$ is served and another part is transmitted to server $j+1$. It can be easily proved that these ``mixed'' scenarios cannot lead to strictly larger worst-case backlog (see Lemma 4 in~\cite{BN15} for a proof).

\begin{lemma}
  \label{lem:cases}
  There exists $k$ such that $B_j = B^k_{j}$ (that is, for
  all $x_{j},\ldots,x_n, x^*$, we have $B_{j}(x_{j},\ldots,x_n,x^*)
  =B^k_{j}(x_{j},\ldots,x_n,x^*)$) and for all $k$, $\xi_j^k \geq
  \xi_j^{k+1}$.
\end{lemma}

\begin{proof}
  The proof is also by induction. We prove that $\forall j$, we have the equivalence 
\begin{equation*}
\begin{split}
  B_j = B_j^k \Leftrightarrow & ~\forall k'>k~\xi_{j+1}^{k'} >
  \frac{\sum_{i\leq j}
    r^*_i+\sum_{\ell>k'}\xi_{j+1}^{\ell}r_j^\ell }{R_j -
    \sum_{\ell=j}^{k'} r_j^\ell} \\ & \text{ and }\xi_{j+1}^k \leq
  \frac{\sum_{i\leq j}
    r^*_i+\sum_{\ell>k}\xi_{j+1}^{\ell}r_j^\ell }{R_j -
    \sum_{\ell=j}^k r_j^\ell}.
\end{split}
\end{equation*}

Assertion {\bf (B)} is proved at the same time: assuming that {\bf
  (B)} is satisfied for server $j+1$, we will prove it for server $j$.

  The equivalence above also state that if for all $k'>k$
  $B_j\neq B^{k'}_j$, then $B_j^{n} \leq B_j^{n-1}
  \leq \cdots \leq B_j^{k}$. 

  Indeed, we have the equivalence $B_j \neq B_j^n \Leftrightarrow
  \xi_{j+1}^n > \frac{ r^*_j}{R_j - \sum_{\ell=j}^n
    r_j^\ell}\Leftrightarrow \frac{r^*_j}{R_j - \sum_{\ell=j}^n  r_j^\ell}<\frac{ r^*_j + \xi_{j+1}^{n} r_j^n}{R_j - \sum_{\ell=j}^{n-1}
    r_j^\ell} $

But 
$$
B_j^n = C_j +x^* + \sum_{\ell \leq n}  \frac{\sum_{i\leq j} r^*_i}{R_j - \sum_{\ell=j}^n
    r_j^\ell} Q_j^\ell $$  
and 
$$ 
B_j^{n-1} = C_j + x^* + \sum_{\ell \leq n-1}  \frac{\sum_{i\leq j} r^*_i + \xi_{j+1}^{n} r_j^\ell}{R_j - \sum_{\ell=j}^{n-1} r_j^\ell} Q_j^\ell + \xi_{j+1}^n Q_j^n.
$$
 So $B_j^n \leq B_j^{n-1}$. The next steps can be shown similarly. 

 Let us assume that $B_j^{n} \leq B_j^{n-1} \leq \cdots \leq B_j^{k}$.

  The coefficient of $Q^\ell_j$ in $B_j^{k'}$ for all $k'\leq k$
  and $\ell >k$ is $\xi_{j+1}^{\ell}$ and that of $Q^k_j$ is
\begin{itemize}
\item either
  $\frac{ r^*_j+\sum_{\ell>k}\xi_{j+1}^{\ell}r_j^\ell }{R_j -
    \sum_{\ell=j}^{k} r_j^\ell} $ (for
  $B_j^k$)
\item or $\xi_{j+1}^k$ (for $B_j^{\ell}$, $\ell<k$).
\end{itemize}

Suppose that 
\begin{equation}
\label{eq:delta}
\frac{    r^*_j+\sum_{\ell>k}\xi_{j+1}^{\ell}r_j^\ell }{R_j -
    \sum_{\ell=j}^{k} r_j^\ell} \geq \xi_{j+1}^k. 
\end{equation}

For $\ell \leq k' <k$, the coefficient of $Q^{\ell}_j$ in $B_j^{k'}$ is
 
\begin{align*}
    \frac{r^*_j+\sum_{\ell>k'}\xi_{j+1}^{\ell}r_j^\ell }{R_j -
    \sum_{\ell=j}^{k'} r_j^\ell } & \leq   \frac{ 
    r^*_j+ \sum_{\ell' =k+1}^n r_j^{\ell'}\xi_{j+1}^{\ell'} + \sum_{\ell' = k'+1}^k
    r_j^{\ell'}\xi_{j+1}^{k}} {R_j - \sum_{\ell'=j}^{k'} r_j^{\ell'}} \\
   &\leq  \frac{
    r^*_j+ \sum_{\ell' =k+1}^n r_j^{\ell'}\xi_{j+1}^{\ell'} + \sum_{\ell' = k'+1}^k
    r_j^{\ell'}\frac{
    r^*_j+\sum_{\ell>k}\xi_{j+1}^{\ell}r_j^\ell }{R_j -
    \sum_{\ell=j}^{k} r_j^\ell}} {R_j - \sum_{\ell'=j}^{k'} r_j^{\ell'}} \\
&= \frac{ 
    r^*_j+ \sum_{\ell' =k+1}^n r_j^{\ell'}\xi_{j+1}^{\ell'}}{R_j - \sum_{\ell'=j}^{k'} r_j^{\ell'}} 
  \left(1+ \frac{\sum_{\ell' = k'+1}^k
    r_j^{\ell'}}{R_j - \sum_{\ell=j}^{k} r_j^\ell}\right)\\
&=\frac{ 
    r^*_j+ \sum_{\ell' =k+1}^n r_j^{\ell'}\xi_{j+1}^{\ell'}}{R_j - \sum_{\ell=j}^{k} r_j^\ell}
\end{align*}

which is the coefficient of $Q^{\ell}_j$ in $B_j^{k}$. The first
inequality uses $({\bf B})$ for server $j+1$ and the second uses
Inequality~\eqref{eq:delta}.

Moreover, for $k' < \ell < k$, the coefficient of $Q^{\ell}_j$ in $B_j^{k'}$ is  $$\xi_{j+1}^{k'} \leq \xi_{j+1}^k \leq \frac{
    r^*_j+\sum_{\ell>k}\xi_{j+1}^{\ell}r_j^\ell }{R_j -
    \sum_{\ell=j}^{k} r_j^\ell},$$ 
so $B_j^k \geq
B_j^{k'}$ and $B_j = B_j^k $.

The same kind of computations lead to $B_j^{k-1} > B_j^k $
if $\xi_{j+1}^k > \frac{
    r^*_j+\sum_{\ell>k}\xi_{j+1}^{\ell}r_j^\ell }{R_j -
    \sum_{\ell=j}^{k} r_j^\ell}$.

Set $\xi_j^k$ according to the $B_j$ that is computed (that is $B_j^k$ for some $k$). Then we still have $\xi_j^\ell \geq \xi_j^{\ell+1}$ and ${\bf (B)}$ is true for server $j$. 
\end{proof}

Finally, we have $$B = \sum_{j\leq n} ( r^*_j T_j+ \sum_{k\geq j} \xi_j^k r^k_j T_j)  + \sum_{i\leq n} b^*_i + \sum_{j\leq k\leq n} \xi_j^k b_j^k.$$

\subsection{Adaptation to trees}
Consider now a sink-tree, the above analysis is still valid, and each
branch of the tree can be analysed independently: if a server has
several predecessors, then the optimization will be for each of the on
disjoint sets of servers and flows. 

The algorithm still runs in polynomial time.

\end{document}